\documentclass[12pt, reqno]{amsart}

\makeatletter
\g@addto@macro{\endabstract}{\@setabstract}
\makeatother

\usepackage{graphics, stackrel}
\usepackage{amsmath, amssymb, amsthm}
\usepackage{graphicx}
\usepackage{verbatim}
\usepackage{amsfonts}
\usepackage{filecontents}
\usepackage{natbib}
\usepackage{fancyvrb}
\usepackage{color}

\usepackage[citecolor=blue, colorlinks=true, linkcolor=blue]{hyperref}

\usepackage{mathrsfs}
\usepackage{bbm}

\usepackage[left=1.25in, right=1.25in, top=1.0in, bottom=1.15in, includehead, includefoot]{geometry}

\usepackage{multirow}

\newcommand{\st}{\ensuremath{\ \mathrm{s.t.}\ }}

\newcommand{\fore}{\therefore \quad}

\newcommand{\tod}{\stackrel { d } {\to} }

\newcommand{\iidsim}{\stackrel {\textrm{ {\sc iid }}} {\sim} }
\newcommand{\1}{\mathbbm 1}

\newcommand*\diff{\mathop{}\!\mathrm{d}}
\renewcommand{\epsilon}{\varepsilon}
\renewcommand{\phi}{\varphi}

\newcommand{\cC}{\mathscr C}

\newcommand{\bB}{\mathscr B}

\newcommand{\hH}{\mathscr H}

\newcommand{\fF}{\mathscr F}

\newcommand{\ZZ}{\mathsf Z}
\renewcommand{\SS}{\mathsf S}
\newcommand{\TT}{\mathsf T}
\newcommand{\CC}{\mathsf C}
\newcommand{\DD}{\mathsf D}

\newcommand{\pP}{\mathscr P}

\newcommand{\RR}{\mathbbm R}

\newcommand{\NN}{\mathbbm N}
\newcommand{\PP}{\mathbbm P}

\newcommand{\EE}{\mathbbm E}

\theoremstyle{plain}

\newtheorem{theorem}{Theorem}[section]

\newtheorem{lemma}{Lemma}[section]
\newtheorem{proposition}{Proposition}[section]

\theoremstyle{definition}

\newtheorem{example}{Example}[section]
\newtheorem{remark}{Remark}[section]

\newtheorem{assumption}{Assumption}[section]

\DeclareMathOperator{\diag}{diag}
\renewcommand{\underline}[1]{\text{\b{$#1$}}} % underbar similar to \bar{}
  
\usepackage{caption}

\usepackage[flushleft]{threeparttable}

\newcommand{\me}{\mathrm{e}}

\usepackage{ dsfont }

\newcommand{\vertiii}[1]{{\left\vert\kern-0.25ex\left\vert\kern-0.25ex\left\vert #1 
    \right\vert\kern-0.25ex\right\vert\kern-0.25ex\right\vert}}
    
\usepackage{subcaption}
\usepackage{graphicx}  

\makeatletter
\newcommand*{\rom}[1]{\expandafter\@slowromancap\romannumeral #1@}
\makeatother

\usepackage{stackengine}

\begin{document}

\title{}

\date{\today}

\begin{center}
  \LARGE 
  The Income Fluctuation Problem and the \\ Evolution of Wealth\footnote{We
      thank the editors and two anonymous referees for many valuable comments
      and suggestions. This paper has also benefited from discussion with many
      colleagues.  We particularly thank Fedor Iskhakov, Larry Liu and Chung
      Tran for their insightful feedback and suggestions. The second author
      gratefully acknowledges financial support from ARC grant FT160100423.\\ 
      \emph{Email addresses:}
      \texttt{qingyin.ma@cueb.edu.cn}, \, \texttt{john.stachurski@anu.edu.au}, \, \texttt{atoda@ucsd.edu}.
  } 
 
  \vspace{1em}
  \normalsize
  Qingyin Ma\textsuperscript{a}, \, John Stachurski\textsuperscript{b} \, and
  \, Alexis Akira Toda\textsuperscript{c} \par 
  
  \vspace{1em}

  \textsuperscript{a}International School of Economics and Management, \\
  Capital University of Economics and Business\\
  \textsuperscript{b}Research School of Economics, Australian National University\\
  \textsuperscript{c}Department of Economics, University of California San Diego

  \vspace{1em}
  %\normalsize{\today}
  January 30, 2020
\end{center}

\begin{abstract} 
    We analyze the household savings problem in a general setting where
    returns on assets, non-financial income and impatience are all state
    dependent and fluctuate over time.  All three processes can be serially
    correlated and mutually dependent.  Rewards can be bounded or unbounded
    and wealth can be arbitrarily large.  Extending classic results from an
    earlier literature, we determine conditions under which (a) solutions
    exist, are unique and are globally computable, (b) the resulting
    wealth dynamics are stationary, ergodic and geometrically mixing, and (c) the wealth distribution has a Pareto tail.  We show how
    these results can be used to extend recent studies of the
    wealth distribution. Our conditions have natural economic interpretations
    in terms of asymptotic growth rates for discounting and return on savings.

    \vspace{1em}
    
    \noindent
    %\textit{JEL Classifications:} XX, YY \\
\end{abstract}

%\maketitle

%section
\section{Introduction}

It has been observed that, in the US and several other large economies, the wealth distribution 
is heavy tailed and wealth inequality has
risen sharply over the last few decades.\footnote{For example, in a study based on capital income data, \cite{saez2016wealth} find that, in the case of the US, 
	the share of total household wealth held by the top 0.1\% increased from 7 percent to 22 percent between 1978 and 2012.  For a discussion of the heavy-tailed property of the wealth distribution, see \cite{Pareto1896LaCourbe}, \cite{davies2000distribution}, \cite{BenhabibBisin2018}, \cite{vermeulen2018fat} or references therein.}
This matters not only for its
direct impact on taxation and redistribution policies, but also for potential
flow-on effects for productivity growth, business cycles and fiscal policy, as
well as for the political environment that shapes these and other economic
outcomes.\footnote{One analysis of the two-way interactions between inequality
    and political decision making can be found in
    \cite{acemoglu2002political}.  \cite{glaeser2003injustice} show how
    inequality can alter economic and social outcomes through subversion of
    institutions.  The same study contains references on linkages between
    inequality and growth.  Regarding fiscal policy, \cite{brinca2016fiscal}
    find strong correlations between wealth inequality and the magnitude of
    fiscal multipliers, while \cite{bhandari2018inequality} study the
connection between fiscal-monetary policy, business cycles and inequality.
\cite{ahn2018inequality} discuss the impact of distributional properties on
macroeconomic aggregates.}

At present, our understanding of these phenomena is hampered by the fact that 
standard tools of analysis---such as those used for heterogeneous agent
models---are not well adapted to studying the wealth distribution as it
stands.  For example, while we have sound
understanding of the household problem when returns on savings and rates of
time discount are constant (see, e.g., \cite{schechtman1976income},
\cite{SchechtmanEscudero1977}, \cite{deaton1992behaviour},
\cite{carroll1997buffer}, or \cite{accikgoz2018existence}), our knowledge is far more
limited in settings where these values are stochastic.  This is problematic,
since injecting such features into the household problem is
essential for accurately representing the joint distribution of income and wealth (e.g.,
\cite{benhabib2015wealth}, \cite{BenhabibBisinLuo2017},
\cite{stachurski2019impossibility}).\footnote{Also related is the
    recent experimental study of \cite{epper2018time}, which finds a strong positive connection between dispersion 
    in subjective rates of time discounting across the population and realized dispersion in the wealth
    distribution.  This in turn is consistent with earlier empirical studies such as
\cite{lawrance1991poverty}.}
    Moreover, models with
time-varying discount rates and returns on assets are at the forefront of recent quantitative analysis of wealth
and inequality.\footnote{For a recent quantitative study see, for example,
    \cite{hubmer2018comprehensive}, where returns on savings and discount
    rates are both state dependent (as is labor income).  
    \cite{kaymak2018accounting} find that asset return heterogeneity is required to match the
    upper tail of the wealth distribution.  %Earlier quantitative studies of the wealth distribution that include stochastic returns on savings can be found in \cite{quadrini2000entrepreneurship}, \cite{krebs2003-RED,krebs2003-QJE,krebs2006}, \cite{angeletos2005incomplete,angeletos-calvet2006}, \cite{cagetti2006entrepreneurship} and \cite{angeletos2007uninsured}. By introducing heterogeneity across agents in returns to capital or savings, these studies generate skewed wealth distributions similar to those observed in US data.
}

While it might be hoped that the analysis of the income fluctuation problem
(or household consumption and savings problem) 
changes little when we shift from constant to state dependent asset returns
and rates of time discount, this turns out not to be the case.  Effectively modeling these features and the way they map to
the wealth distribution requires significant advances in our understanding of
choice and stochastic dynamics in the setting of optimal savings.  

One difficulty is that state-dependent discounting takes us beyond the bounds
of traditional dynamic programming theory.  This matters little if there exists some
constant $\bar \beta < 1$ such that the discount process $\{\beta_t\}$
satisfies $\beta_t \leq \bar \beta$ for all $t$ with probability one, since,
in this case, a standard contraction mapping argument can still be applied
(see, e.g., \cite{miao2006competitive} or \cite{cao2020recursive}).  However, recent quantitative studies
extend beyond such settings.  For example, AR(1) specifications are
increasingly common, in which case the support  of
$\beta_t$ is unbounded above at every point in time.\footnote{See, for example, \cite{hills2018fiscal}, \cite{hubmer2018comprehensive}
or \cite{schorfheide2018identifying}.}
    Even if discretization
is employed, the outcome $\beta_t \geq 1$ can occur with positive
probability when the approximation is sufficiently fine.  Moreover, such outcomes are not inconsistent with
empirical and experimental evidence,
at least for some households in some states of the world.\footnote{See, for example, \cite{loewenstein1991negative} and \cite{loewenstein1991workers}.}  Do there exist
conditions on $\{\beta_t\}$ that allow for $\beta_t \geq 1$ in some states and
yet imply existence of optimal polices and practical computational techniques?

Another source of complexity for the income fluctuation problem in the general
setting considered here is that the set of possible values for household
assets is typically unbounded above.   For example, when returns on assets are
stochastic, a sufficiently long sequence of favorable returns can compound one
another to project a household to arbitrarily high levels of wealth.  This
model feature is desirable: We wish to analyze these kinds of outcomes rather
than rule them out.  Indeed, \cite{benhabib2015wealth} and other related
studies argue convincingly that such outcomes are a key causal mechanism behind
the heavy tail of the current distribution of wealth.\footnote{One related
    study is \cite{benhabib2011distribution}, who show that capital income
    risk is the driving force of the heavy-tail properties of the stationary
    wealth distribution. In Blanchard-Yaari style economies,
    \cite{Toda2014JET}, \cite{TodaWalsh2015JPE} and \cite{benhabib2016distribution} show that
    idiosyncratic investment risk generates a double Pareto stationary wealth
    distribution.  \cite{gabaix2016dynamics} point out that a positive
correlation of returns with wealth (``scale dependence'') in addition to
persistent heterogeneity in returns (``type dependence'') can well explain the
speed of changes in the tail inequality observed in the data.} However, if we
accept this logic, then stationarity and ergodicity of the wealth
process---which are fundamental both for estimation and for simulation-based
numerical methods---must now be established in a setting where the wealth
distribution has unbounded support.  In such a scenario, what conditions on
preferences and financial and labor income are necessary for these properties
to hold?

A final and related example of the need for deeper analysis is as follows: To
understand the upper tail of the wealth distribution, we must avoid
unnecessarily truncating the upper tail of the set of possible asset values in
quantitative work.  While truncation is convenient because finite or compact
state spaces are easier to handle computationally, we can attain greater
accuracy in modeling the wealth distribution if truncation at the upper tail
can be replaced locally by a parameterized savings function, such as a linear
function \citep*{gouin2018pareto}.  However, any such approximation must be
justified by theory.  What conditions can be imposed on primitives to generate
such properties while still maintaining realistic assumptions for asset
returns and non-financial income?

In this paper we address all of these questions, along with other key
properties of the income fluctuation problem, such as continuity and
monotonicity of the optimal consumption policy. Our setting admits capital
income risk, labor earnings shocks and time-varying discount rates, 
driven by a combination of {\sc iid} innovations and an
exogenous Markov chain $\{ Z_t \}$.  The supports of the innovations can be unbounded,
so we admit practical innovation sequences such as normal and lognormal.
As a whole, this environment allows for a range of realistic features, such as
stochastic volatility in returns on asset holdings, or correlation in the
shocks impacting asset returns and non-financial income. The utility function can be
unbounded both above and below, with no specific structure imposed beyond
differentiability, concavity and the usual slope (Inada)
conditions.\footnote{While the assumption that the exogenous state process
    $\{Z_t\}$ is a (finite state) Markov chain might appear
restrictive, it fits most practical settings and avoids a host of technical issues
that tend to obscure the key ideas.  Moreover, the innovation shocks are not restricted
to be discrete, and the same is true for assets and consumption.}

To begin, when considering optimality in the household problem, we require a
condition on the state dependent discount process $\{\beta_t\}$ that
generalizes the classical condition $\beta < 1$ from the constant case and,
for reasons discussed above, permits $\beta_t > 1$ with positive probability.  
To this end, we introduce the restriction\footnote{Here and below we set $\beta_0 \equiv 1$, so $\prod_{t=1}^n \beta_t = \prod_{t=0}^n \beta_t$.}
\begin{equation}
    \label{eq:fpbc}
    G_\beta < 1
    \quad \text{where} \quad 
    G_\beta := \lim_{n \to \infty} 
    \left(\EE \prod_{t=1}^n \beta_t \right)^{1/n}.
\end{equation}
Condition \eqref{eq:fpbc} clearly generalizes the classical condition $\beta <
1$ for the constant discount case.  In the stochastic case,
$\ln G_\beta$ can be understood as the asymptotic growth rate of the
probability weighted average discount factor.  Indeed, if 
$B_n := \EE \prod_{t=1}^n \beta_t$ is the average $n$-period discount factor,
then, from the definition of $G_\beta$ and some straightforward analysis, we
obtain $\ln (B_{n+1}/B_n) \to \ln G_\beta$,
so the condition $G_\beta < 1$ implies that the asymptotic growth rate of 
the average $n$-period discount factor is negative, drifting down from its
initial condition $\beta_0 \equiv 1$ at the rate $\ln G_\beta$.  This does not, of course, preclude the
possibility that $\beta_t > 1$ at any given $t$.

We show that condition \eqref{eq:fpbc} is in fact a necessary condition in those
settings where the classical condition is necessary for finite lifetime
values.  In this sense it cannot be further weakened for the income
fluctuation problem apart from special cases.  At the same time, it
admits the use of convenient specifications such as the discretized AR(1)
process from \cite{hubmer2018comprehensive}.  In addition, we prove that
$G_\beta$ can be represented as the spectral radius of a nonnegative matrix,
and hence can be computed by numerical linear algebra (as discussed below).

We also generalize the standard condition $\beta R < 1$, where $R$ is the
gross interest rate in the constant case, which is used to ensure stability of
the asset path and finiteness of lifetime valuations, as well as existence of
stationary Markov policies (see, e.g., \cite{deaton1992behaviour},
\cite{chamberlain2000optimal} or \cite{li2014solving}).
Analogous to \eqref{eq:fpbc}, we introduce the generalized condition
\begin{equation}
    \label{eq:fpbrc}
    G_{\beta R} < 1
    \quad \text{where} \quad 
    G_{\beta R} := \lim_{n \to \infty} 
        \left(\EE \prod_{t=1}^n \beta_t R_t \right)^{1/n} .
\end{equation}
Here $\{R_t\}$ is a stochastic capital income process.  Analogous to the case
of $G_\beta$, the value $\ln G_{\beta R}$ can be understood as the asymptotic
growth rate of average gross payoff on assets, discounted to present value.

We show that, when Conditions~\eqref{eq:fpbc}--\eqref{eq:fpbrc} hold and
non-financial income satisfies two moment conditions, a unique optimal
consumption policy exists.  We also show that the policy can be computed by
successive approximations and analyze its properties, such as monotonicity and
asymptotic linearity. This asymptotic linearity can be used to successfully
model wealth inequality by accurately representing asset path dynamics for
very high wealth households \citep*{gouin2018pareto}.

One important feature of Conditions \eqref{eq:fpbc}--\eqref{eq:fpbrc} is that
they take into
account the autocorrelation structure of preference shocks and asset returns.
For example, if these processes depend only on {\sc iid} innovations, then
\eqref{eq:fpbc} reduces to $\EE \beta_t < 1$ and \eqref{eq:fpbrc} reduces to $\EE \beta_t R_t < 1$.  But returns on assets are
typically not {\sc iid}, since both mean returns and volatility are, in
general, time varying, and preference shocks are typically modeled as
correlated (see, e.g., \cite{hubmer2018comprehensive} or \cite{schorfheide2018identifying}).
This dependence must be and is accounted for in \eqref{eq:fpbrc}, since long
upswings in $\{\beta_t\}$ and $\{R_t\}$ can lead to explosive paths for valuations and assets.

Next we study asymptotic stability, stationarity and ergodicity of wealth.
Such properties
are essential to existence of stationary equilibria in heterogeneous agent
models (e.g., \cite{huggett1993risk}, \cite{aiyagari1994uninsured} or
\cite{cao2020recursive}), as well as
standard estimation, calibration and simulation techniques that connect
time series averages with cross-sectional moments.\footnote{A well-known example of
    a computational technique that uses ergodicity can be
    found in \cite{krusell1998income}. On the estimation side see, for
    example, \cite{hansen2002generalized}.}
These properties require an additional restriction, placed on the asymptotic growth rate
of mean returns.  Analogous to \eqref{eq:fpbc} and
\eqref{eq:fpbrc}, this is defined as
\begin{equation}
    \label{eq:gr} 
    G_R := \lim_{n \to \infty} \left( \EE \prod_{t=1}^n R_t \right)^{1/n} .
\end{equation}
We show that if $G_R$ is sufficiently restricted and a degree of social mobility
is present, then there exists a unique stationary
distribution for the state process, the distributional path of the state
process under the optimal path converges globally to the stationary
distribution, and the stationary distribution is ergodic.  
We also show that, under some mild additional conditions, the rate of
convergence of marginal distributions to the stationary distribution is
geometric, and that a version of the Central Limit
Theorem is valid.  Finally, under some mild additional conditions, we prove that the stationary distribution of
assets is Pareto tailed, consistent with the data.

Our study is related to \cite{benhabib2015wealth}, who prove the
existence of a heavy-tailed wealth distribution in an infinite horizon
heterogeneous agent economy with capital income risk.  In the process, they show that households
facing a stochastic return on savings possess a unique optimal consumption
policy characterized by the (boundary constraint-contingent) Euler equation,
and that a unique and unbounded stationary distribution exists for wealth
under this consumption policy.  They assume isoelastic utility, constant
discounting, and mutually independent, {\sc iid} returns and labor income
processes, both supported on bounded closed intervals with strictly
positive lower bounds.  We relax all of these assumptions.  Apart from allowing more
general utility and state dependent discounting, this permits such realistic
features for household income as positive correlations between labor earnings and wealth returns
(an extension that was suggested by \cite{benhabib2015wealth}), or time
varying volatility in returns.\footnote{Empirical motivation for these kinds
    of extensions can be found in numerous studies, including 
\cite{guvenen2014inferring} and \cite{fagereng2016heterogeneityNBER, fagereng2016heterogeneityAERPP}.}

Another related paper is \cite{chamberlain2000optimal}, which studies an
income fluctuation problem with stochastic income and asset returns and
obtains many significant results on asymptotic properties of consumption.
Their study imposes relatively few restrictions on the wealth return and labor
income processes.  Our paper extends their work by allowing for random
discounting, as well as dropping their boundedness restriction on the utility,
which prevents their work from being used in many standard settings such as constant relative risk aversion.
We also develop a set of new results on stability and ergodicity, as well as
asymptotic normality of the wealth process.

Our optimality theory draws on techniques found in \cite{li2014solving}, who
show that the time iteration operator is a contraction mapping with respect to
a metric that evaluates consumption differences in terms of marginal utility,
while assuming a constant discount factor and constant rate of return on
assets.\footnote{\cite{Coleman1990} introduced the time iteration operator as a
    constructive method for solving stochastic growth models. It has since been used in \cite{DattaMirmanReffett2002}, \cite{MorandReffett2003} and many other
studies.}   We show that these ideas extend to a setting where both returns and
discount rates are stochastic and time varying.  Our results on dynamics under
the optimal policy have no counterparts in \cite{li2014solving}.

In a similar vein, our work is related to several other papers that treat the
standard income fluctuation problems with constant rates of return on assets
and constant discount rates, such as \cite{rabault2002borrowing},
\cite{carroll2004theoretical} and \cite{kuhn2013recursive}.  While
\cite{carroll2004theoretical} constructs a weighted supremum norm contraction
and works with the Bellman operator, the other two papers focus on time
iteration. In particular, \cite{rabault2002borrowing} exploits the
monotonicity structure, while \cite{kuhn2013recursive} applies a version of
the Tarski fixed point theorem.  Our techniques for studying optimality are
close to those in \cite{li2014solving}, as discussed
above.\footnote{Our paper is also related to \cite{cao2017persistent}, who study
    wealth inequality in a continuous-time framework with heterogeneous
    returns following a two-state Markov chain. While we do not pursue the
    connection here, the generality of our setup,
    including a persistent shock structure to wealth returns, might
    permit a study of the continuous-time limit that yields the tail
    results of \cite{cao2017persistent} in a general framework.}

The rest of this paper is structured as follows. Section \ref{s:ifp}
formulates the problem and establishes optimality results. Sufficient
conditions for the existence and uniqueness of optimal policies are discussed.
Section \ref{s:sto_stability} focuses on stochastic stability. 
Section~\ref{s:gc} discusses our key conditions and how they can be checked.
Section \ref{s:app} provides a set of applications and Section~\ref{s:c}
concludes. All proofs are deferred to the appendix.  Code that generates our figures can be found at
\url{https://github.com/jstac/ifp_public}.

\section{The Income Fluctuation Problem and Optimality Results}
\label{s:ifp}

This section formulates the income fluctuation problem we consider,
establishes the existence, uniqueness and computability of a solution, and
derives its properties.

\subsection{Problem Statement}
\label{ss:ifp_problem}

We consider a general income fluctuation problem,
where a household chooses a consumption-asset path $\{(c_t, a_t)\}$ to solve
\begin{align}
    \label{eq:trans_at}
    & \max \, \EE_0 \left\{ 
                \sum_{t = 0}^\infty 
                \left(\prod_{i=0}^t \beta_i \right) u(c_t)
             \right\}    \nonumber \\
  \st \quad  
    & a_{t+1} = R_{t+1} (a_t - c_t) + Y_{t+1},  \\
    & 0 \leq   \; c_t \leq a_t, \quad (a_0, Z_0)=(a,z) \text{ given} \nonumber.
\end{align}
Here  $u$ is the utility function, $\{\beta_t\}_{t \geq 0}$ is discount factor
process with $\beta_0=1$, $\{R_t\}_{t \geq 1}$ is the gross rate of return on wealth,
and $\{Y_t \}_{t \geq 1}$ is non-financial income. These stochastic processes obey
\begin{equation}
    \label{eq:RY_func}
      \beta_t = \beta \left( Z_t, \epsilon_t \right), 
      \quad
      R_{t} = R \left( Z_{t}, \zeta_{t} \right),
      \quad \text{and} \quad
      Y_{t} = Y \left( Z_{t}, \eta_{t} \right),
\end{equation}
where $\beta$, $R$ and $Y$ are measurable nonnegative functions and
$\{Z_t\}_{t \geq 0}$ is an irreducible time-homogeneous $\ZZ$-valued Markov chain taking
values in finite set $\ZZ$.  Let $P(z, \hat z)$ be the probability of
transitioning from $z$ to $\hat z$ in one step.
The innovation processes $\{\epsilon_t\}$, $\{\zeta_t\}$ and $\{\eta_t\}$ are
{\sc iid} independent and their supports can be continuous and vector-valued.

The function $u$ maps $\RR_+$ to $\{ - \infty \} \cup \RR$, is twice
differentiable on $(0, \infty)$, satisfies $u' > 0$ and $u'' < 0$ everywhere
on $(0, \infty)$, and that $u'(c) \to \infty$ as $c \to 0$ and $u'(c) < 1$
as $c \to \infty$. We define
\begin{equation}
\label{eq:nota1}
  \EE_{a,z} 
      := \EE \left[ \,\cdot \, \big| \, (a_0,Z_0) = (a, z) \right]
  \quad \text{and} \quad
  \EE_z 
      := \EE \left[ \, \cdot \, \big| \, Z_0 = z \right].
\end{equation}
The next period value of a random variable $X$ is typically denoted
$\hat{X}$.  Expectation without a subscript refers to the stationary process,
where $Z_0$ is drawn from its (necessarily unique) stationary distribution.

%\begin{example}
    %The optimization problem stated above includes the problem faced by
    %households in \cite{cao2020recursive}, which in turn builds on
    %\cite{krusell1998income}.  In that model, randomness in the discount
    %factor process is driven by idiosyncratic preference shocks $\{i_t\}$, while
    %returns on assets fluctuate due to an aggregate shock process $\{s_t\}$ impacting
    %productivity.  Non-financial income is affected by both the idiosyncratic 
    %shocks and the aggregate shocks.  His scenario fits our
    %framework when $Z_t = (s_t, i_t)$.  The innovation vectors in \eqref{eq:RY_func} are not
    %required.
%\end{example}

\subsection{Key Conditions}

\label{ss:fss}

Our conditions for optimality are listed below.
In what follows, $G_\beta$ is
the asymptotic growth rate of the discount process as defined in
\eqref{eq:fpbc}.

\begin{assumption}
	\label{a:b0}
     The discount factor process satisfies $G_\beta < 1$.
\end{assumption}

Assumption~\ref{a:b0} is a natural extension of the standard
condition $\beta < 1$ from the constant discount case.  
If $\beta_t \equiv \beta$ for all
$t$, then $G_\beta = \beta$, as follows immediately from the
definition.  It is weaker than the obvious sufficient condition
$\beta_t \leq \bar \beta$ with probability one for some constant
$\bar \beta < 1$, since in such a setting we have $G_\beta \leq \bar \beta <
1$.   In fact it cannot be significantly weakened, as the
proposition shows.  

\begin{proposition}
    [Necessity of the discount condition]
    \label{p:necg}
    Let $\beta_t$ and $u(Y_t)$ be positive with probability one for all $t$
    and all initial states $z$ in $\ZZ$.  If, in this setting,
    we have $G_\beta \geq 1$, then the objective in \eqref{eq:trans_at} is
    infinite at every initial state $(a, z)$.
\end{proposition}

The positivity assumed here may or may not hold in applications, but
Proposition~\ref{p:necg} shows that special conditions will have to be imposed on
preferences if Assumption~\ref{a:b0} fails.  Put differently, allowing
$G_\beta \geq 1$ is tantamount to allowing $\beta \geq 1$ in the case when the
discount rate is constant.

Next, we need to ensure that the present discounted value of wealth
does not grow too quickly, which requires a joint restriction on asset returns and
discounting.  When $\{R_t\}$ and $\{ \beta_t\}$ are constant at
values $R$ and $\beta$, the standard restriction from the existing literature
is $\beta R < 1$.  A generalization
using $G_{\beta R}$ as defined in \eqref{eq:fpbrc} is 

\begin{assumption}
	\label{a:rb0}
     The discount factor and return processes satisfy $G_{\beta R} < 1$.
\end{assumption}

Finally, we impose routine technical restrictions on non-financial income.
The second restriction is needed to exploit first order conditions.

\begin{assumption}
    \label{a:y0}
    $\EE \, Y < \infty$ and $\EE \, u'(Y) < \infty$.
\end{assumption}

Next we provide one example where Assumptions~\ref{a:b0}--\ref{a:y0} are
easily verified.  More complex examples are deferred to Sections~\ref{s:gc} and \ref{s:app}.

\begin{example}
    \label{eg:bh1}
    Suppose, as in \cite{benhabib2015wealth}, that there is a constant discount
    factor $\beta < 1$, utility is CRRA with $\gamma \geq 1$, 
    $\left\{ R_{t} \right\}$ and $\left\{ Y_{t} \right\}$ are {\sc iid}, mutually
    independent, supported on bounded closed intervals of strictly positive real
    numbers, and, moreover,
    \begin{equation}
        \label{eq:bhc}
        \beta \EE R_t^{1 - \gamma} < 1
        \quad \text{and} \quad
        ( \beta \EE R_t^{1 - \gamma} )^{1/\gamma} \EE R_t < 1.
    \end{equation}
    Assumptions~\ref{a:b0}--\ref{a:y0} are all satisfied in this case.  To see
    this, observe that $G_\beta = \beta <1$ in the constant discount case, so
    Assumption~\ref{a:b0} holds. Since $x \mapsto x^{1-\gamma}$ is convex when
    $\gamma \geq 1$, Jensen's inequality implies that $\EE R_t^{1-\gamma}\geq (\EE
    R_t)^{1-\gamma}$. Multiplying both sides of the last inequality by $\beta
    (\EE R_t)^\gamma$ yields 
    \begin{equation*} 
        G_{\beta R}=\beta \EE R_t =
        \beta (\EE R_t)^{1-\gamma}(\EE R_t)^\gamma \leq (\beta \EE
    R_t^{1-\gamma})(\EE R_t)^\gamma .
    \end{equation*} 
    By the second condition of \eqref{eq:bhc}, Assumption~\ref{a:rb0} holds.
    Assumption~\ref{a:y0} also
holds because $Y_t$ is restricted to a compact subset of the positive reals.
\end{example}

\subsection{Optimality: Definitions and Fundamental Properties} 

To consider optimality, we temporarily assume that $a_0>0$ and
set the asset space to $(0, \infty)$.\footnote{Assumption~\ref{a:y0} combined
with $u'(0)= \infty$  implies that $\PP \{ Y_t >
0 \}=1$ for all $t \geq 1$. Hence, $\PP \{ a_t > 0 \} = 1$ for all $t \geq 1$
  and excluding zero from the asset space makes no difference to optimality.}
The state space for $\{(a_t, Z_t) \}_{t \geq 0}$ is then
  $\SS_0:= (0, \infty) \times \ZZ$.
A \emph{feasible policy} is a Borel measurable function 
$c \colon \SS_0 \to \RR$ with $0 \leq c(a,z) \leq a$ for all 
$(a,z) \in \SS_0$. A feasible policy $c$ and initial condition 
$(a,z) \in \SS_0$ generate an asset path 
$\{ a_t\}_{t \geq 0}$ via \eqref{eq:trans_at} when 
$c_t = c (a_t, Z_t)$ and $(a_0, Z_0) = (a,z)$. The lifetime value of 
policy $c$ is
\begin{equation}
    V_c (a,z) = \EE_{a,z} \sum_{t = 0}^\infty 
        \beta_0 \cdots \beta_t u \left[ c (a_t, Z_t) \right], \label{eq:Vc}
\end{equation}
where $\{ a_t\}$ is the asset path generated by $(c,(a,z))$. In the Appendix 
we show that $V_c$ is well-defined on $\SS_0$.
A feasible policy $c^*$ is called \emph{optimal} if $V_c \leq V_{c^*}$ on
$\SS_0$ for any feasible policy $c$. 
A feasible policy is said to satisfy the \textit{first order optimality  condition} if 
\begin{equation}
\left( u'\circ c \right)(a,z) \geq 
    \EE_{z} \, \hat{\beta} \hat{R} 
              \left( u' \circ c \right)
              \left( 
                  \hat{R} \left[ a - c(a,z) \right] + \hat{Y}, 
                  \, \hat{Z} 
              \right)
\end{equation}
for all $(a,z) \in \SS_0$, and equality holds when 
$c(a,z) < a$. Noting that $u'$ is decreasing, the first order optimality condition can be compactly stated as
    \begin{equation}
	\label{eq:foc}
	    \left( u' \circ c \right)(a,z) = 
	    \max \left\{ 
	            \EE_z \, \hat{\beta} \hat{R} 
	                      \left( u' \circ c \right) 
	                        \left( \hat{R} \left[ a - c(a,z)\right] + \hat{Y}, 
	                               \, \hat{Z} 
	                        \right),
	            u'(a)
	         \right\}
	\end{equation}
for all $(a,z) \in \SS_0$. A feasible policy is said to satisfy the 
\textit{transversality condition} if, for all $(a, z) \in \SS_0$,
\begin{equation}
\label{eq:tvc}
  \lim_{t \to \infty}     
    \EE_{a,z} \, \beta_0 \cdots \beta_t 
            \left(u' \circ c \right)(a_t, Z_t) \, a_t = 0.
\end{equation}

\begin{theorem}[Sufficiency of first order and transversality conditions]
    \label{t:opt_result}
    If Assumptions \ref{a:b0}--\ref{a:y0} hold, then every feasible
    policy satisfying the first order and transversality conditions is an
    optimal policy.
\end{theorem}

\subsection{Existence and Computability of Optimal Consumption}
\label{ss:ifp_existence}

Let $\cC$ be the space of continuous functions $c
\colon \SS_0 \to \RR$ such that
$c$ is increasing in the first argument,
$0 < c(a,z) \leq a$ for all $(a,z) \in \SS_0$, and
\begin{equation}
    \label{eq:C4}
   \sup_{(a,z) \in \SS_0} \left| (u' \circ c)(a,z) - u'(a) \right| < \infty.
\end{equation}
To compare two consumption policies, we pair $\cC$ with the distance
\begin{equation}
\label{eq:rho_metric}
  \rho(c,d) 
    := \left\| u' \circ c - u' \circ d \right\|
    := \sup_{(a,z) \in \SS_0} 
             \left| 
                 \left(u' \circ c \right)(a,z) - 
                 \left(u' \circ d \right)(a,z) 
             \right|,
\end{equation}
which evaluates the maximal difference in terms of marginal utility. While
elements of $\cC$ are not generally bounded, $\rho$ is a
valid metric on $\cC$. In particular, $\rho$ is finite on $\cC$ since 
$\rho(c,d) \leq \left\| u' \circ c - u' \right\| 
              + \left\| u' \circ d - u' \right\|$, 
and the last two terms are finite by \eqref{eq:C4}.
In Appendix \ref{s:proof_opt}, we show
that $(\cC, \rho)$ is a complete metric space.  The following proposition
shows that, for any policy in $\cC$, the first order optimality condition \eqref{eq:foc} implies the
transversality condition.

\begin{proposition}[Sufficiency of first order condition]
    \label{pr:suff_optpol}
    Let Assumptions~\ref{a:b0}--\ref{a:y0} hold.
    If $c \in \cC$ and the first order optimality condition \eqref{eq:foc}
    holds for all $(a,z) \in \SS_0$, then $c$ satisfies the transversality
    condition. In particular, $c$ is an optimal policy.
\end{proposition}

We aim to characterize the 
optimal policy as the fixed point of the \emph{time iteration operator} $T$ defined 
as follows: for fixed $c \in \cC$ and $(a,z) \in \SS_0$, the value of the 
image $Tc$ at $(a,z)$ is defined as the $\xi \in (0,a]$ that solves
\begin{equation}
\label{eq:T_opr}
    u'(\xi) = \psi_c(\xi, a, z),
\end{equation}
where $\psi_c$ is the function on 
\begin{equation}
\label{eq:dom_T_opr}
    G := \left\{ 
            (\xi, a, z) \in \RR_+ \times (0, \infty) \times \ZZ \colon
            0 < \xi \leq a
         \right\}
\end{equation}
defined by
\begin{equation}
\label{eq:keypart_T_opr}
    \psi_c(\xi,a,z) := 
      \max \left\{
              \EE_{z} \, \hat{\beta} \hat{R}
                 (u' \circ c)[\hat{R}(a - \xi) + \hat{Y}, \, \hat{Z}], 
              \, u'(a)
           \right\}.
\end{equation}

The following theorem shows that the time iteration
operator is an $n$-step contraction mapping on a complete metric space of
candidate policies and its fixed point is the unique optimal policy.

\begin{theorem}[Existence, uniqueness and computability of optimal policies]
	\label{t:ctra_T}
    If Assumptions~\ref{a:b0}--\ref{a:y0} hold, then there exists an
    $n$ in $\NN$ such that $T^n$ is a contraction mapping on $(\cC, \rho)$.
	In particular,
	\begin{enumerate}
	    \item\label{i:ctra_T1} $T$ has a unique fixed point $c^* \in \cC$.
	    \item\label{i:ctra_T2} The fixed point $c^*$ is the unique optimal policy in $\cC$.
        \item\label{i:ctra_T3} For all $c \in \cC$ we have $\rho(T^k c, c^*)
            \to 0$ as $k \to \infty$.
	\end{enumerate}
\end{theorem}

Part (\ref{i:ctra_T3}) shows that, under our conditions, the familiar time iteration
algorithm is globally convergent, provided one starts with some policy in the candidate
class $\cC$.

\subsection{Properties of Optimal Consumption}
\label{ss:ifp_properties}

In this section we study the properties of the optimal consumption function
obtained in Theorem~\ref{t:ctra_T}.
Assumptions~\ref{a:b0}--\ref{a:y0} are held to be true throughout.
The following two propositions show the monotonicity of the consumption
function, which is intuitive.

\begin{proposition}[Monotonicity with respect to wealth]
    \label{pr:monotonea}
    The optimal consumption and savings functions $c^*(a,z)$ and $i^*(a,z) := a - c^*(a,z)$ are increasing in $a$.
\end{proposition}

\begin{proposition}[Monotonicity with respect to income]
    \label{pr:monotoneY} If $\{ Y_{1t} \}$ and $\{ Y_{2t} \}$ are two income processes
    satisfying $Y_{1t}\leq Y_{2t}$ for all $t$ and $c_1^*$ and $c_2^*$ are the
    corresponding optimal consumption functions, then $c_1^* \leq c_2^*$
    pointwise on $\SS_0$.
\end{proposition}

Under further assumptions we can show that the optimal policy is concave and asymptotically linear with respect to the wealth level.

\begin{proposition}[Concavity and asymptotic linearity of consumption function]
	\label{pr:optpol_concave}
	If for each $z \in \ZZ$ and $c \in \cC$ that is concave in its first argument,
	\begin{equation}
		\label{eq:concave_prop}
		x \mapsto (u')^{-1} \left[ 
		\EE_z \hat \beta \hat{R} \left( u' \circ c \right) 
		(\hat{R} x + \hat{Y}, \, \hat{Z} )
		\right]
		\; \text{ is concave on } \RR_+,
	\end{equation}
	then
	\begin{enumerate}
		\item $a \mapsto c^*(a,z)$ is concave, and
        \item there exists $\alpha(z) \in [0,1]$ such that $\lim_{a \to
            \infty} [c^*(a,z) / a] = \alpha(z)$.
	\end{enumerate}
\end{proposition}

\begin{remark}
	\label{rm:concave}
	Condition \eqref{eq:concave_prop} imposes some concavity structure on utility.
    It holds for the constant relative risk aversion (CRRA) utility function
    \begin{equation}
	    \label{eq:crra_utils}
	    u(c) = \frac{c^{1 - \gamma}}{1 - \gamma}
	    \quad \text{if } \gamma > 0 
	    \quad \text{and} \quad
	    u(c) = \log c 
	    \quad \text{if } \gamma = 1,
    \end{equation}
    as shown in Appendix~\ref{s:proof_opt}.
\end{remark}

Proposition~\ref{pr:optpol_concave} states that $c^*(a, z) \approx \alpha(z) a
+ b(z)$ for some function $b(z)$ when $a$ is large.  This provides
justification for linearly extrapolating the policy functions when computing
them at high wealth levels. 

Together, parts (1) and (2) of Proposition~\ref{pr:optpol_concave} imply the
linear lower bound $c^*(a,z) \geq \alpha(z)a$, although they do not provide a concrete
number for $\alpha(z)$. The following proposition establishes an explicit
linear lower bound.

\begin{proposition}[Linear lower bound on consumption]
	\label{pr:optpol_linbound}
    If there exists a nonnegative constant $\bar s$ such that 
    \begin{equation}
        \label{eq:suff_linbound}
        \bar s < 1 
        \qquad \text{and} \qquad 
        \EE_z \, \hat{\beta} \hat{R} \, u' (\hat{R} \, \bar s \, a) \leq u'(a) \text{ for all } (a,z) \in \SS_0,
    \end{equation}
    then $c^*(a,z) \geq (1-\bar s) a$ for all $(a,z) \in \SS_0$.\footnote{We
        adopt the convention $0 \cdot \infty = 0$, so condition
        \eqref{eq:suff_linbound} does not rule out the case $\PP \{R_t =0 \mid
        Z_{t-1} = z\} > 0$. Indeed, as shown in the proofs, the conclusions
    still hold if we replace this condition by the weaker alternative $\EE_z
\hat{\beta} \hat{R} \, u'[ \hat{R} \bar s a + (1 - \bar s) \hat{Y}]  \leq u'(a)$
for all $(a,z) \in \SS_0$.}
\end{proposition}

The second inequality in \eqref{eq:suff_linbound} restricts marginal utility derived from
transferring wealth to the next period and then consuming versus consuming
wealth today.  The value $\bar s$ can be clarified once primitives are specified,
as the next example illustrates.

\begin{example}
    \label{eg:CRRA}
    Suppose that utility is CRRA, as in \eqref{eq:crra_utils}.  If we now take
    \begin{equation}
        \label{eq:suff_crra}
        \bar s := 
        \left(
            \max_{z \in \ZZ} \EE_z \hat \beta \hat R^{1 - \gamma} 
        \right)^{1 / \gamma}
    \end{equation}
    and $\bar s < 1$, then the conditions of
    Proposition~\ref{pr:optpol_linbound} hold.  In particular, the second
    inequality in \eqref{eq:suff_linbound} holds, as follows directly
    from the definition of $\bar s$ and $u'(x) = x^{-\gamma}$.
    In the case of \cite{benhabib2015wealth}, where the discount rate is
    constant and returns are {\sc iid}, the expression in \eqref{eq:suff_crra}
    reduces to $\bar s := (\beta \EE R_t^{1 - \gamma} )^{1 / \gamma}$.
    The requirement $\bar s < 1$ then reduces to $\beta
    \EE R_t^{1 - \gamma} < 1$, which is one of their
    assumptions (see Example~\ref{eg:bh1}).  
\end{example}

\section{Stationarity, Ergodicity, and Tail Behavior}
\label{s:sto_stability}

This section focuses on stationarity, ergodicity and tail behavior of wealth under the unique optimal policy $c^*$ obtained in Theorem~\ref{t:ctra_T}.  
So that this policy exists, Assumptions~\ref{a:b0}--\ref{a:y0} are always
taken to be valid. We extend $c^*$ to $\SS$ by setting $c^* (0,z) = 0 $ for
all $z \in \ZZ$ and consider dynamics of $(a_t, Z_t)$ on $\SS := \RR_+ \times
\ZZ$, the law of motion for which is
\begin{subequations}
    \begin{align}
    \label{eq:dyn_sys}
      a_{t+1} &= R \left( 
                      Z_{t+1}, \zeta_{t+1} 
                   \right) 
                 \left[ a_t - c^* \left(a_t, Z_t \right) \right]
                 + Y \left( 
                        Z_{t+1}, \eta_{t+1} 
                     \right), 
                      \\
      %\tag{DS1}
      Z_{t+1} &\sim P \left( Z_t, \, \cdot \, \right)
    \end{align}
\end{subequations}
Let $Q$ be the joint stochastic kernel of $(a_t, Z_t)$ on $\SS$. See
Appendix~\ref{s:prelm} for this and related definitions.

\subsection{Stationarity}
To obtain existence of a stationary distribution we need to restrict the
asymptotic growth rate for asset returns $G_R$ defined in 
\eqref{eq:gr}.

\begin{assumption}
    \label{a:r0}
    There exists a constant $\bar s$ such that \eqref{eq:suff_linbound} holds
    and $\bar s \, G_R < 1$.
\end{assumption}

Below is one straightforward example of a setting where this holds, with more complex applications
deferred to Sections~\ref{s:gc}--\ref{s:app}.

\begin{example}
    \label{eg:bh2}
    Assumption~\ref{a:r0} holds in the setting of
    \cite{benhabib2015wealth}.  As shown in Example~\ref{eg:CRRA},
    with $\bar s := (\beta \EE R_t^{1 - \gamma} )^{1 / \gamma}$ and the
    assumptions of \cite{benhabib2015wealth} in force, the conditions of \eqref{eq:suff_linbound} hold. 
    Moreover, in their {\sc iid} setting we have $G_R = \EE R_t$, so $\bar s \, G_R < 1$
    reduces to  $(\beta \EE R_t^{1 - \gamma})^{1/\gamma} \EE R_t < 1$. This
    is one of their conditions, as discussed in Example~\ref{eg:bh1}.
\end{example}

By Proposition~\ref{pr:optpol_linbound}, the value $\bar s$ in
Assumption~\ref{a:r0} is an upper bound on the rate of savings.  $G_R$ is an
asymptotic growth rate for each unit of savings invested.  If the product
of these is less than one, then probability mass contained in the wealth
distribution will not drift to
$+\infty$, which allows us to obtain the following result.\footnote{Assumption~\ref{a:r0} is weaker than any restriction implying wealth is
    bounded from above---a common device for compactifying the state space and
    thereby obtaining a stationary distribution.  Indeed, under many
    specifications of $\{Y_t\}$ and $\{R_t\}$ that fall within our framework,
    wealth of a given household can and will, over an infinite horizon, exceed any
    finite bound with probability one.  See, for example,
\cite{benhabib2015wealth}, Proposition~6.}

\begin{theorem}
    [Existence of a stationary distribution]
	\label{t:sta_exist}
    If Assumption~\ref{a:r0} holds, then $Q$ admits at least one stationary
    distribution on $\SS$.
\end{theorem}

Stationarity of the form obtained in Theorem~\ref{t:sta_exist} is required to
establish existence of stationary recursive equilibria in 
heterogeneous agent models with idiosyncratic risk, such as
\cite{huggett1993risk} or \cite{aiyagari1994uninsured}.\footnote{For models
    with aggregate shocks, such as \cite{krusell1998income}, a fully specified recursive equilibrium requires
    that households take the wealth distribution as one component of the
    state in their savings problem, and that stationarity holds for the entire
    joint distribution (defined over a product space encompassing both the
    wealth distribution and the exogenous state process).  These
    problems fall outside the scope of
    Theorem~\ref{t:sta_exist}, since $\{Z_t\}$ is finite-valued.
    For a careful treatment of stationary recursive equilibrium
in Krusell--Smith type models, see \cite{cao2020recursive}.}

\subsection{Ergodicity}
While Assumption~\ref{a:r0} implies existence of a stationary distribution, it
is not in general sufficient for uniqueness or stability.  For these additional properties to hold, we must impose
sufficient mixing.
In doing so, we consider the following two cases:
\begin{enumerate}
    \item[(Y1)] The support of $\{Y_t\}$ is finite.
	\item[(Y2)] The process $\{Y_t\}$ admits a density representation.
\end{enumerate}

Condition (Y2) means that there exists a function $f$ from $\RR_+ \times \ZZ$
to $\RR_+$ such that
\begin{equation}
    \label{eq:yden}
	\PP \{ Y_t \in A \mid Z_t = z \} 
	= \int_A f (y \mid z) \diff y
\end{equation}
for all Borel sets $A \subset \RR_+$ and all $z$ in $\ZZ$.

\begin{assumption}
	\label{a:pos_dens}
    There exists a $\bar z$ in $\ZZ$ such that $P(\bar z, \bar z) > 0$. Moreover,
    with $y_\ell \geq 0$ defined as the greatest lower bound of the support of $\{Y_t\}$,  
    either
    \begin{itemize}
        \item (Y1) holds and $\PP\{Y_t = y_\ell \mid Z_t = \bar z\} > 0$, or
        \item (Y2) holds and there exists a $\delta > y_\ell$ such that 
            $f \left( \cdot \mid \bar z \right) > 0$ on $(y_\ell, \delta)$.
    \end{itemize}
\end{assumption}

Assumption~\ref{a:pos_dens} requires that there is a positive probability of
receiving low labor income at some relatively persistent state of the world
$\bar{z}$.  This is a mixing condition that enforces social mobility.  The
reason is that $\{Z_t\}$ is already assumed to be irreducible, so $\bar z$ is
eventually visited by each household.  For any such household, there is a
positive probability of low labor income over a long period. Wealth then
declines.  In other words, currently rich households or dynasties will not be
rich forever.  This guarantees sufficient social mobility between rich and
poor, generating ergodicity.

To state our uniqueness and stability results, let $Q^t$
be the $t$-step stochastic kernel, let $\| \cdot \|_{TV}$ be total
variation norm and let 
	$V(a,z) := a + m_V$,
where $m_V$ is a constant to be defined in the proof. For any integrable real-valued
function $h$ on $\SS$, 
let
\begin{equation*}
    \bar h(a, z) := h(a, z) - \EE h(a_t, Z_t)
\end{equation*}
and 
\begin{equation*}
    \gamma_h^2 
    := \EE \left[ \, \bar{h}^2 (a_0, Z_0) \right] + 
                2 \sum_{t=1}^{\infty} 
                 \EE \left[ \, 
                     \bar{h}(a_0, Z_0) \bar{h}(a_t, Z_t) 
                     \right],
\end{equation*}
where, here and in the theorem below, $\EE$ indicates expectation under
stationarity.

\begin{theorem}
    [Uniqueness, stability, ergodicity and mixing]
	\label{t:gs_gnl_ergo_LLN}
    If Assumptions~\ref{a:r0} and \ref{a:pos_dens} hold, then
	\begin{enumerate}
        \item the stationary distribution $\psi_\infty$ of $Q$ is unique and
            there exist constants $\lambda < 1$ and $M < \infty$ such that,
            \begin{equation*}
                \left\| Q^t \left( (a,z), \cdot \right) - \psi_\infty \right\|_{TV}
                \leq \lambda^t M V(a, z)
                \quad \text{for all } (a, z) \in \SS.
            \end{equation*}
		\item For all $(a,z) \in \SS$ and real-valued function $h$ on $\SS$
            such that $\EE |h(a_t, Z_t)| < \infty$, 
		\begin{equation*}
			\PP_{a,z}
			\left\{
			\lim_{T \to \infty} 
                \frac{1}{T} 
                \sum_{t=1}^T h(a_t, Z_t) = \EE h(a_t, Z_t)  
			\right\}
			= 1.
		\end{equation*}
		\item $Q$ is $V$-geometrically mixing. Moreover, if $\gamma_h^2  > 0$ and $h^2 / V$ is bounded, 
		\begin{equation*}
            \frac{1}{\sqrt{T \gamma_h^2}  }
            \sum_{t=1}^{T} \bar h(a_t, Z_t) 
			 \tod \, N(0,1)
			\quad \text{as } \, T \to \infty.
		\end{equation*}
	\end{enumerate} 
\end{theorem}

Part 1 of Theorem~\ref{t:gs_gnl_ergo_LLN} states that the stationary
distribution $\psi_\infty$ is unique and asymptotically attracting at a geometric
rate.  Part 2 states that the state process is ergodic, and hence long-run
sample moments for individual households coincide with cross-sectional
moments.  The notion of mixing discussed in Part~3 is defined in the appendix.
It states that social mobility holds asymptotically and mixing occurs at a
geometric rate, although the rate may be arbitrarily slow.
This mixing is enough to provide a Central Limit Theorem for the state
process, which is the second claim in Part~3.

\subsection{Tail Behavior}

Having established the stationarity and ergodicity of wealth, we now study the tail behavior of the wealth distribution. We show that the wealth distribution is either bounded or (unbounded and) heavy-tailed under mild conditions. To prove this result we introduce the following assumption.

\begin{assumption}
	\label{a:wealth_growth}
    The assumptions of Proposition \ref{pr:optpol_concave} are satisfied, so
    the optimal policy $a \mapsto c^*(a,z)$ is concave and asymptotically
    linear: $\lim_{a \to \infty} c^*(a,z)/a = \alpha(z)\in [0,1]$.
    Furthermore, there exists $\bar z \in \ZZ$ such that $P(\bar z, \bar z) >
    0$ and
    \begin{equation}
        \label{eq:wealth_growth}
        \PP_{\bar z} \{ R(\bar z,\hat{\zeta})(1-\alpha(\bar z)) > 1 \} > 0.
    \end{equation}
\end{assumption}

\begin{remark}
    Condition \eqref{eq:wealth_growth} implies that wealth grows with nonzero
    probability when it is large. Indeed, using the law of motion
    \eqref{eq:dyn_sys} and noting that $Y\geq 0$, if $Z_t=Z_{t+1}=\bar z$, then by
    \eqref{eq:wealth_growth} we have
    \begin{equation*}
    \frac{a_{t+1}}{a_t} \geq R \left( \bar z, \zeta_{t+1} \right) 
                     \left[ 1 - c^*(a_t, \bar z )/a_t \right] > 1
    \end{equation*}
    with positive probability if $a_t$ is large enough.
\end{remark}

To state our result on tail behavior, we introduce the following notation. For
any nonnegative function $A(z,\hat{z},\hat{\zeta})$, define the $\ZZ \times
\ZZ$ matrix-valued function $M_A$ by
\begin{equation}\label{eq:MGF}
(M_A(s))(z,\hat{z})=\EE_{z,\hat{z}} A(z,\hat{z},\hat{\zeta})^s.
\end{equation}
Elements of $M_A(s)$ are conditional moment generating functions of $\log A$.
In the statement below, $\odot$ denotes the Hadamard (entry-wise) product, and
$r(\cdot)$ returns the spectral radius of a matrix.
Also $a_\infty$ is a random variable with distribution $\psi_\infty(\cdot , \ZZ)$.

\begin{theorem}[Tail behavior]\label{t:heavy_tail}
    Let Assumptions~\ref{a:r0}--\ref{a:wealth_growth} hold and define
    \begin{subequations}
        \begin{align}
        G(z,\hat{z},\hat{\zeta})&=R(\hat{z},\hat{\zeta})(1-\alpha(z)), \label{eq:defG}\\
        A(z,\hat{z},\hat{\zeta})&=G(z,\hat{z},\hat{\zeta}) \1 \{
        G(z,\hat{z},\hat{\zeta}) > 1 \}, \text{ and } \label{eq:defA}\\
        \lambda(s)&=r(P \odot M_A(s)). \label{eq:lambda}
        \end{align}
    \end{subequations}
    Then $\lambda$ is convex in $s \geq 0$.
    Assume that there exists $s>0$ in the interior of the domain of $\lambda$ such that $1<\lambda(s)<\infty$ and let
    \begin{equation}
        \kappa
        :=\inf\{ s>0 \, | \, \lambda(s)>1 \}.\label{eq:kappa}
    \end{equation}
    If $a_\infty$ has unbounded support, then it is
    heavy-tailed. In particular, for any $\epsilon>0$,
    \begin{equation}\label{eq:heavy_tail}
    \liminf_{a\to\infty} a^{\kappa+\epsilon}\PP \{ a_\infty \geq a \}>0.
    \end{equation}
\end{theorem}

\begin{remark}\label{rm:lambda}
The assumption $1 < \lambda(s) < \infty$ for some $s>0$ is weak. Because the
$(\bar{z},\bar{z})$-th element of $P \odot M_A(s)$ is
\begin{equation*}
P(\bar{z},\bar{z})\EE_{\bar{z},\bar{z}}G(\bar z,\bar z,\hat{\zeta})^s \1 \{ G(\bar z, \bar z,\hat{\zeta}) > 1 \},
\end{equation*}
by the definition of $G$ in \eqref{eq:defG} and condition \eqref{eq:wealth_growth}, we always have $\lambda(s)\to\infty$ as $s\to\infty$. Hence there exists $s>0$ such that $\lambda(s)\in (1,\infty)$ if, for example, $\hat{\zeta}$ has a compact support.
\end{remark}

Condition \eqref{eq:heavy_tail} implies that for any $\epsilon>0$, there exists a constant $C(\epsilon)>0$ such that
\begin{equation*}
\PP \{ a_\infty \geq a \} \geq C(\epsilon) a^{-\kappa-\epsilon}
\end{equation*}
for large enough $a$, so the upper tail of the wealth distribution is at least Pareto.

\begin{remark}
\cite{Toda2019JME} constructs an example of a \cite{huggett1993risk} economy with Pareto-tailed wealth distribution when discount factors are random. Theorem \ref{t:heavy_tail} is significantly more general as we allow for stochastic returns and income. \cite{stachurski2019impossibility} prove that with constant discount factor, constant asset return, and light-tailed income, the wealth distribution is always light-tailed. Theorem \ref{t:heavy_tail} shows that sufficient heterogeneity in discount factor or returns generates heavy tails.
\end{remark}

\begin{example}
    \label{eg:bh3}
    The CRRA-{\sc iid} setting of \cite{benhabib2015wealth} satisfies the
    assumptions of Theorem \ref{t:heavy_tail}. When utility is CRRA, by
    Proposition~5 of \cite{benhabib2015wealth}, condition
    \eqref{eq:wealth_growth} holds if $R(\bar z,\hat{\zeta}) > 1/\bar{s}$ with
    positive probability, where $\bar{s}$ is given in Example \ref{eg:CRRA}. In
    the {\sc iid} case, this condition reduces to $\PP \{(\beta \EE
    R_t^{1-\gamma})^{1/\gamma} R_t > 1 \} > 0$, which holds under the conditions
    of \cite{benhabib2015wealth}.\footnote{\cite{benhabib2015wealth} assume that
    $\PP \{ \beta R_t > 1 \} > 0$,
    so it suffices to show that $(\beta
    \EE R_t^{1-\gamma})^{1/\gamma}\ge \beta$ or, equivalently,
    $\EE (\beta R_t)^{1-\gamma} \ge 1$.  By Jensen's inequality and their
    restriction $\gamma \geq 1$, the last bound is true whenever
    $(\EE \beta R_t)^{1-\gamma} \ge 1$.  But this must hold because, under their
    conditions, we have $\beta \EE R_t < 1$, as shown in Example~\ref{eg:bh1}.}
     Thus, Assumption~\ref{a:wealth_growth} holds. The existence
    of $s>0$ with $\lambda(s)\in (1,\infty)$ follows from Remark~\ref{rm:lambda}
    and the assumption that $R_t$ has a compact support.  
\end{example}

\section{Testing the Growth Conditions}

\label{s:gc}

The three key conditions in the paper are the restrictions on the growth rates
$G_\beta$, $G_{\beta R}$ and $G_R$, with the first two required for optimality
and the last for stationarity (see Assumptions~\ref{a:b0}, \ref{a:rb0}
and \ref{a:r0} respectively).  In this section we explore the restrictions
implied by these conditions.  We begin with the following result, which yields
a straightforward method for computing these growth rates.

\begin{lemma}
    [Long-run growth rates and spectral radii]
    \label{l:besr}
    Let $\phi_t = \phi(Z_t, \xi_t)$, where $\phi$ is a nonnegative
    measurable function and $\{\xi_t\}$ is an {\sc iid} sequence with marginal
    distribution $\pi$.  In this setting we have
    \begin{equation}
        \label{eq:deflb}
        G_\phi = r(L_\phi),
        \quad \text{where} \quad
        G_\phi :=
        \lim_{n \to \infty} 
            \left(\EE \prod_{t=1}^n \phi_t \right)^{1/n}
    \end{equation}
    and $r(L_\phi)$ is the spectral radius of the matrix defined by 
    \begin{equation}
        \label{eq:lfunc}
        L_\phi(z, \hat z) 
        = P(z, \hat z) \int \phi(\hat z, \hat \xi) \pi (\diff \hat \xi). 
        %= P(z, \hat z) \EE_{\hat z} \varphi.
    \end{equation}
\end{lemma}

The matrix $L_\phi$ is expressed as a function on $\ZZ \times \ZZ$ in
\eqref{eq:lfunc} but can be represented in traditional matrix notation by
enumerating $\ZZ$.\footnote{Specifically, if  
$\ZZ := \{z_1, \dots, z_N \}$, then
	$L_\varphi = P D_\varphi$
where $P$ is, as before, the transition matrix for the exogenous state, and
$D_\varphi := \diag \left( \EE_{z_1} \varphi, \dots, \EE_{z_N} \varphi
\right)$ when $\EE_{z} \varphi := \EE_z \varphi (z, \hat \xi) $.
In what follows, $D_\beta$, $D_R$ and $D_{\beta R}$ are defined analogously to
$D_\phi$.}

What factors determine
the long-run average growth rates embedded in our assumptions, such as $G_\beta$ or $G_R$?
Lemma~\ref{l:besr} tells us how to compute these values for a given
specification of dynamics, but how should we understand them intuitively and what
factors determine their size?  To address these questions, let us 
consider an AR(1) discount factor process, which has been adopted in
several recent quantitative studies (see, e.g.,
\cite{hubmer2018comprehensive} or \cite{hills2018fiscal}).
In particular, suppose that the state
process follows a discretized version of
\begin{equation}
    \label{eq:zsh}
    Z_{t+1} = (1 - \rho) \mu + \rho Z_t + (1 - \rho^2)^{1/2} \sigma \upsilon_{t+1}, 
    \quad \{\upsilon_t\} \iidsim N(0, 1),
\end{equation}
and $\beta_t = Z_t$.  
(The discretization implies that $\beta_t$ is always positive.)
To simplify interpretation, the process \eqref{eq:zsh} is structured so that the stationary
distribution of $\{Z_t\}$ is $N(\mu, \sigma^2)$.  We use \cite{rouwenhorst1995}'s method to discretize $\{Z_t\}$ 
and then calculate $G_\beta$ using Lemma~\ref{l:besr}, studying
how $G_\beta$ is affected by the parameters in \eqref{eq:zsh}.

Since $\beta_t = Z_t$ for all $t$, the structure of \eqref{eq:zsh} implies
that $\mu$ is the long-run unconditional mean of $\{\beta_t\}$.  It can
therefore be set to standard calibrated value for the discount factor, such as
$0.99$ from \cite{krusell1998income}.  What we wish to understand is how the
remaining parameters $\rho$ and $\sigma$ affect the value of $G_\beta$.  While
no closed form expression is available in this case, Figure~\ref{fig:G_beta}
sheds some light by providing a contour plot of $G_\beta$ over a set of
$(\rho, \sigma)$ pairs.  The figure shows that $G_\beta$ grows with both the
persistence term $\rho$ and volatility term $\sigma$.  In particular, the
condition $G_\beta < 1$ fails when the persistence and volatility of the
discount factor process are sufficiently high.  
This is because
$G_\beta$ is the limit of $\left(\EE \prod_{t=1}^n \beta_t \right)^{1/n}$ and,
for positive random variables, sequence of large outcomes have a strong
compounding effect on their product.  High volatility and high persistence
reinforce this effect.  

%We can gain further intuition on these effects by examining the value of 
%$G_\beta = \lim_{n\to\infty} \left(\EE \prod_{t=1}^n \beta_t \right)^{1/n}$ prior to
%discretizing the AR(1) process $\{Z_t\}$.
%This is in some sense outside the scope of the present paper, since we focus
%on a finite exogenous state process, but is nonetheless valuable because the
%non-discretized process has an analytical solution for $G_\beta$ that closely matches outcomes for the
%discretized version.

\begin{figure} 
    \includegraphics[width=0.8\linewidth]{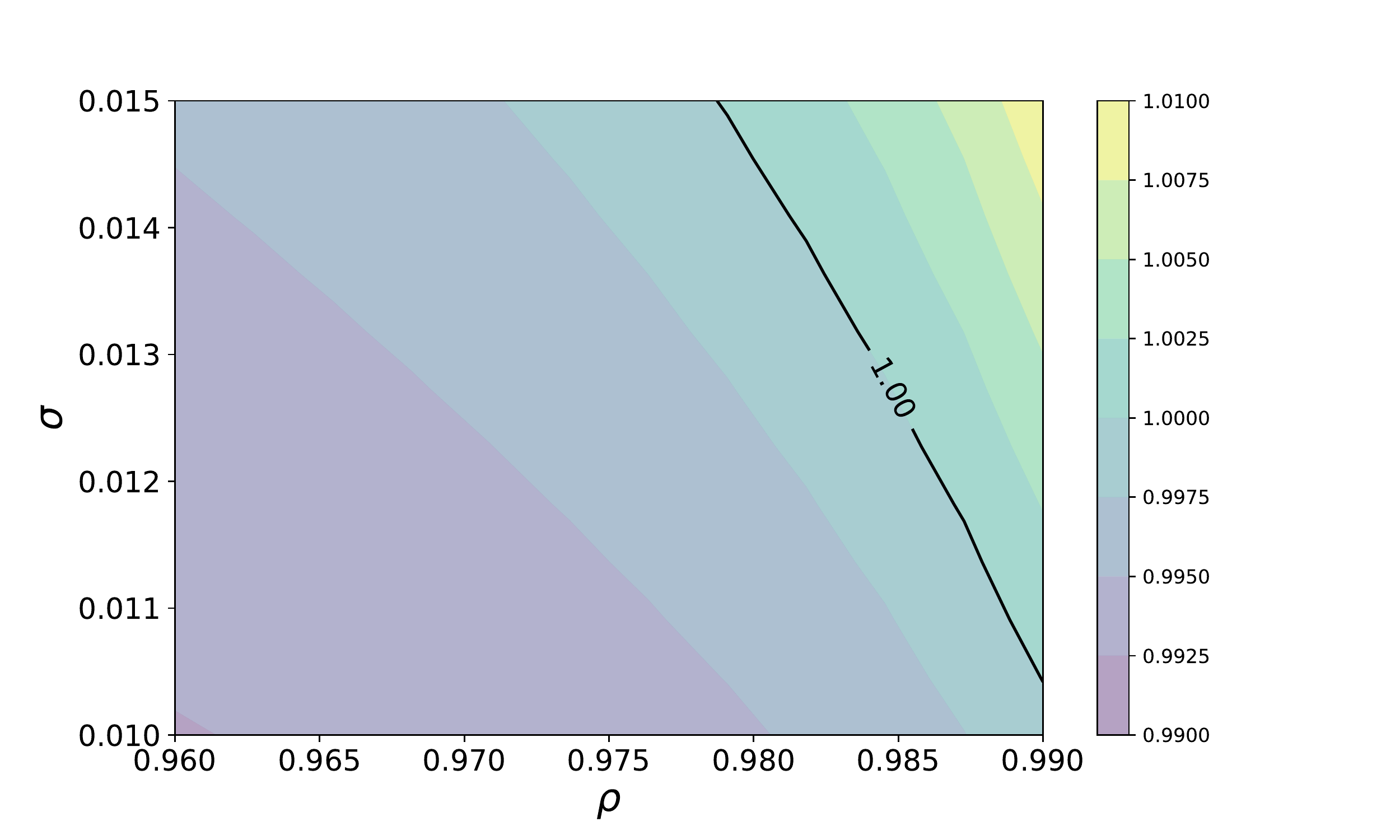}
    \caption{\label{fig:G_beta} Contour plot of $G_\beta$ under AR(1) discounting}
\end{figure}

This discussion has focused on $G_\beta$ but similar intuition applies to both
$G_R$ and $G_{\beta R}$.  If $\beta_t$ and $R_t$ are both increasing functions
of the state process, then these asymptotic growth rates also increase with 
greater persistence and volatility in the state process, as well as higher
unconditional mean.  The next section further illustrates these points.

\section{Application: Stochastic Volatility and Mean Persistence}

\label{s:app}

We showed in Examples~\ref{eg:bh1}, \ref{eg:CRRA} and \ref{eg:bh2} that, in the setting of 
\cite{benhabib2015wealth}, where the discount factor is constant and returns
and labor income are {\sc iid}, Assumptions~\ref{a:b0}--\ref{a:y0} and
Assumption~\ref{a:r0} are all satisfied. Hence,
by Theorems~\ref{t:ctra_T} and \ref{t:sta_exist}, the household optimization 
problem has a unique optimal policy and the wealth process under this policy
has a stationary solution.  If, in addition, the support of $Y_t$ is finite or
$Y_t$ has a positive density, say, then the conditions of
Theorem~\ref{t:gs_gnl_ergo_LLN} also hold and the stationary solution is
ergodic, geometrically mixing and its time series averages are asymptotically
normal.  

Let us now bring the model closer to the data by relaxing the {\sc iid} restrictions on financial and non-financial
returns, introducing both mean persistence and time varying
volatility in returns on assets.\footnote{The importance of these features for wealth dynamics
was highlighted in \cite{fagereng2016heterogeneityNBER}.}
In particular, we set
\begin{equation}
    \label{eq:app_Rt}
    \log R_t = \mu_t + \sigma_t \zeta_t,
\end{equation}
where $\{ \zeta_t\}$ is {\sc iid} and standard normal and $\{\mu_t\}$ and
$\{\sigma_t\}$ are finite-state Markov chains, discretized from 
\begin{equation*}
    \mu_t = (1 - \rho_\mu) \bar{\mu} + \rho_\mu \mu_{t-1} +\delta_\mu \upsilon_t^\mu 
    \quad \text{and} \quad
    \log \sigma_t = (1- \rho_\sigma) \bar{\sigma} 
        + \rho_\sigma \log \sigma_{t-1} + \delta_\sigma \upsilon_t^\sigma .
\end{equation*}
Innovations are {\sc iid} and standard normal.
Using the data in \cite{fagereng2016heterogeneityAERPP} on
Norwegian financial returns over 1993--2003,
we estimate these AR(1) models to obtain
$\bar{\mu} = 0.0281$, $\rho_\mu = 0.5722$, $\delta_\mu = 0.0067$, 
$\bar{\sigma}=-3.2556$, $\rho_\sigma=0.2895$ and $\delta_\sigma=0.1896$. 
Based on this calibration, the stationary mean and standard deviation of  
$\{R_t\}$ are around $1.03$ and $4\%$, respectively.

To distinguish the effects of stochastic volatility and mean persistence,
we consider two subsidiary models. The first reduces 
$\{ \mu_t\}$ to its stationary mean $\bar{\mu}$, while the second reduces 
$\{ \sigma_t \}$ to its stationary mean 
$\tilde{\sigma} := \me^{\bar{\sigma} + \delta_\sigma^2/ 2(1 - \rho_\sigma^2)}$. 
In summary, 
\begin{align*}
    &\log R_t = \bar{\mu} + \sigma_t \zeta_t  
    \qquad (\text{Model \rom{1}})    \\
    &\log R_t = \mu_t + \tilde{\sigma} \zeta_t
    \qquad (\text{Model \rom{2}})
\end{align*}
We set $\beta=0.95$ and $\gamma=1.5$. To test the stability properties of 
Model~\rom{1}, we explore a neighborhood of the calibrated
$(\rho_\sigma, \delta_\sigma)$ values, while in Model~\rom{2}, we 
do likewise for $(\rho_\mu, \delta_\mu)$ pairs. In each scenario, other
parameters are fixed to the benchmark. The results are shown in
Figures~\ref{fig:m1} and \ref{fig:m2}. 

In part~(a) of each figure, we see that $G_{\beta R}$ is increasing in the
persistence and volatility parameters of the state process.  The intuition
behind this feature was explained in Section~\ref{s:gc} for the case of
$G_\beta$ and is similar here. (Note that $G_{\beta R} = \beta G_R$ in the
present case, since $\beta_t \equiv \beta$ is a constant, so $G_{\beta R}$ has
the same shape as $G_R$ in terms of contours.) The dots in the figures
show that $G_{\beta R} < 1$ at the estimated parameter values.

Part~(b) of each figure shows the set of parameters under which the model is
globally stable and ergodic. The stability threshold is the boundary of the
set of parameter pairs that produce $\max \{ G_{\beta R}, \bar s, \bar s G_R
\} < 1$, where $\bar s$ is given by \eqref{eq:suff_crra}.  For such pairs,
Assumptions~\ref{a:rb0} and \ref{a:r0} both hold, 
so the conditions of Theorems~\ref{t:sta_exist}--\ref{t:gs_gnl_ergo_LLN} are satisfied.
(We are continuing to suppose that 
$Y_t$ is finite or has a positive density, so that Assumption~\ref{a:pos_dens}
holds.  Assumptions~\ref{a:b0} and
\ref{a:y0} are always valid in the current setting).  Observe that the estimated
parameter values (dot points) lie inside the stable set. 

%Hence, both models are globally stable and the stationary
%wealth distributions can be computed by Theorem~\ref{t:gs_gnl_ergo_LLN}.
%Moreover, the models are stable for a broad range of parameter values,
%including those corresponding to highly persistent and volatile $\{R_t\}$
%processes. For example, although the estimated parameter values are $(0.2895,
%0.1896)$ in Model I, the economy is globally stable even when $\rho_\sigma$ is
  %higher than 0.95, or when $\delta_\sigma$ is close to $1$. 

\begin{figure}
\centering
\begin{subfigure}[a]{0.8\textwidth}
   \includegraphics[width=1\linewidth]{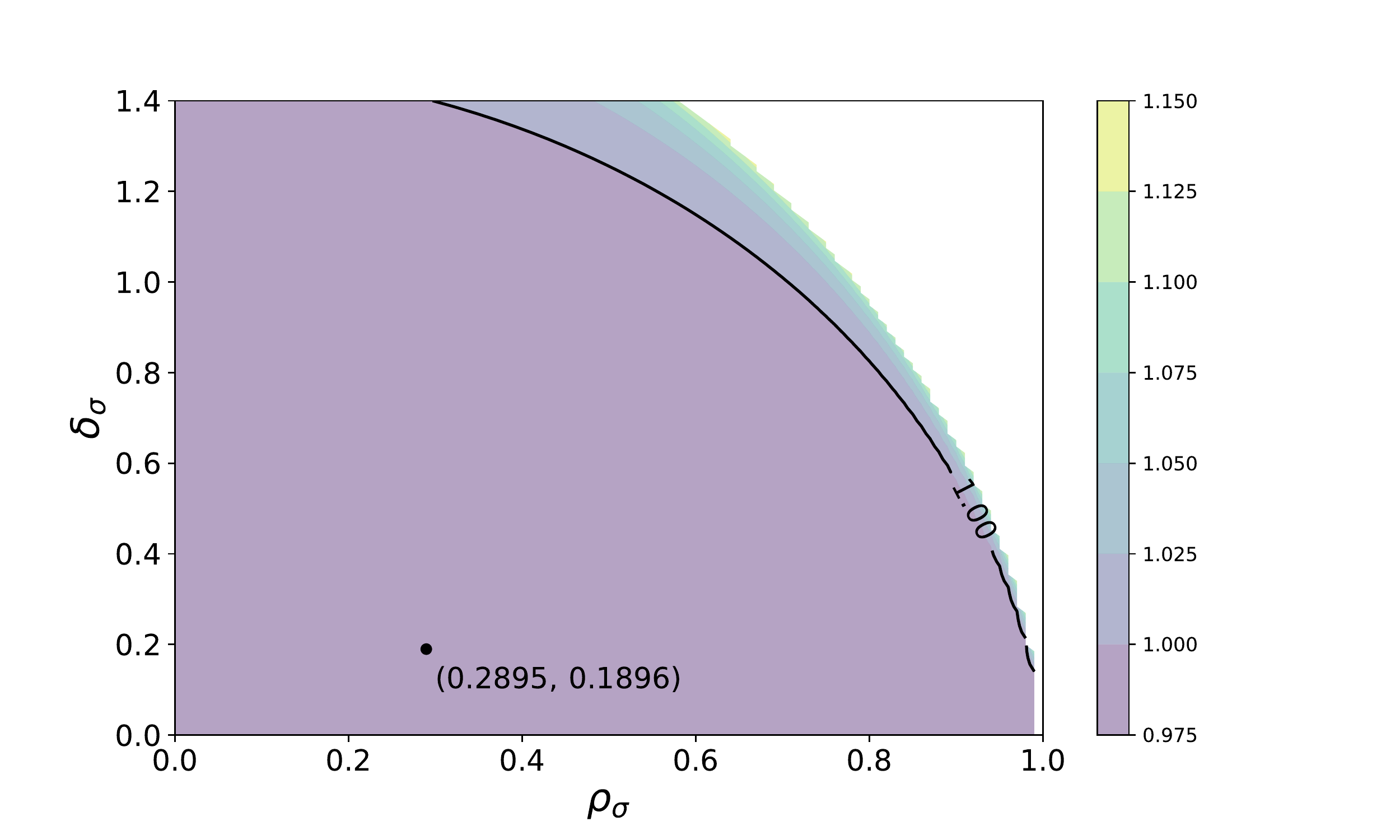}
   \caption*{(a) Contour plot of $G_{\beta R}$}
   \label{fig:m1_G_betR} 
\end{subfigure}
\vspace{0.cm}
\begin{subfigure}[b]{0.7\textwidth}
   \includegraphics[width=1\linewidth]{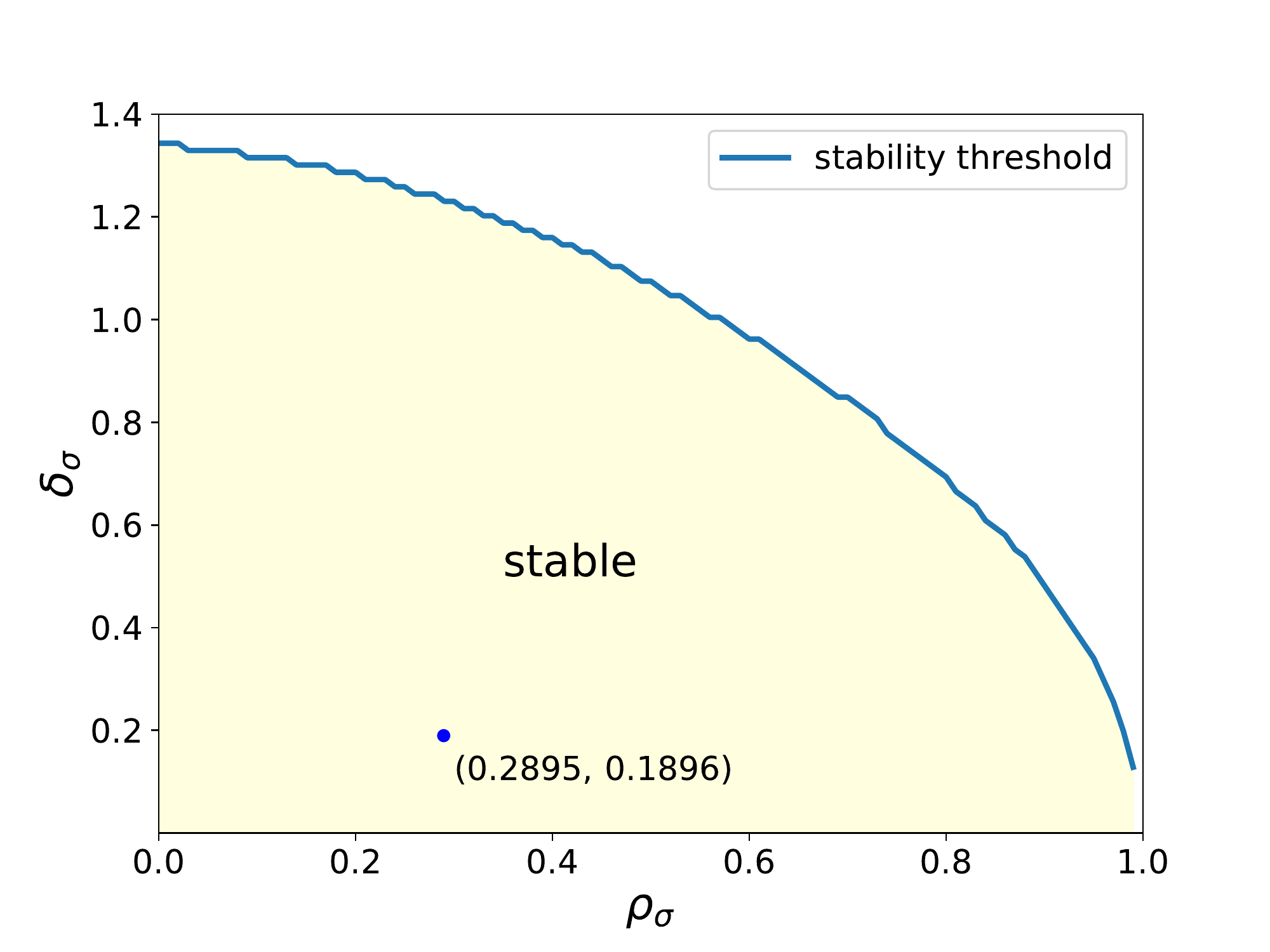}
   \caption*{(b) Range and threshold of stability}
   \label{fig:m1_stb}
\end{subfigure}
\caption[]{Stability tests for Model \rom{1}}
\label{fig:m1}
\end{figure}

%$\beta = 0.95$, \, $\gamma = 2$, \, $\bar{\mu}=0.0281$

\begin{figure}
\centering
\begin{subfigure}[a]{0.8\textwidth}
   \includegraphics[width=1\linewidth]{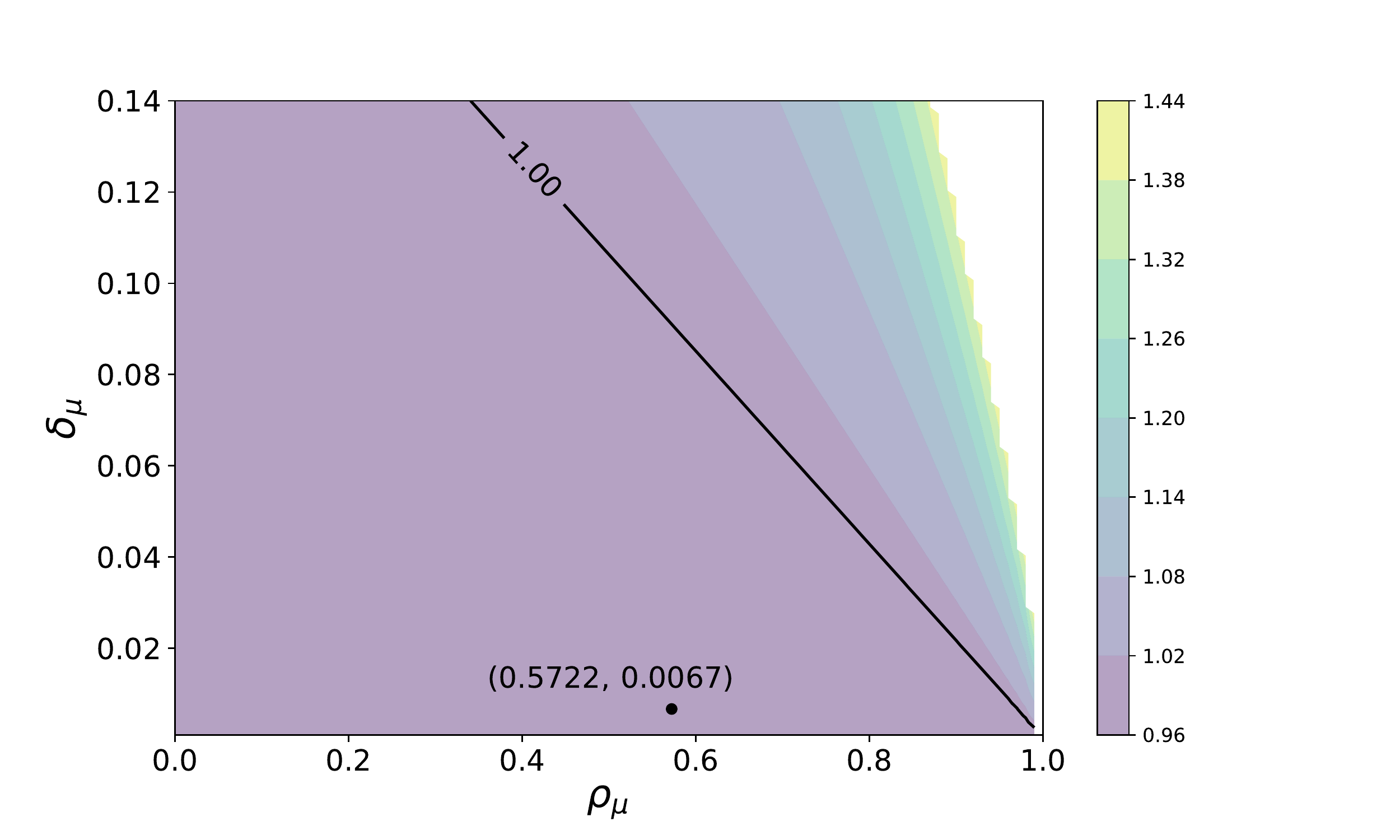}
   \caption*{(a) Contour plot of $G_{\beta R}$}
   \label{fig:m2_G_betR} 
\end{subfigure}
\vspace{0.cm}
\begin{subfigure}[b]{0.7\textwidth}
   \includegraphics[width=1\linewidth]{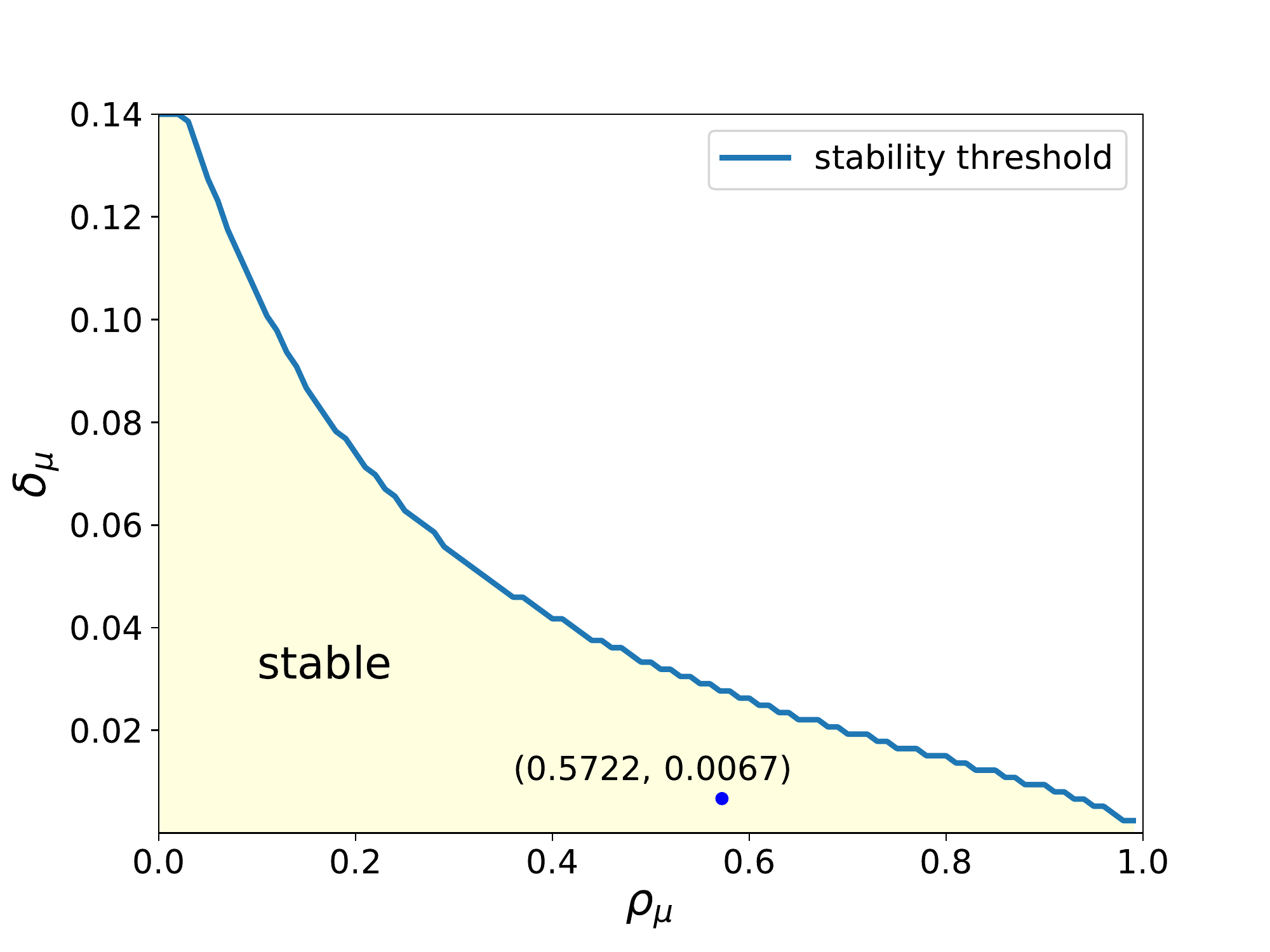}
   \caption*{(b) Range and threshold of stability}
   \label{fig:m2_stb}
\end{subfigure}
\caption[]{Stability tests for Model \rom{2}}
\label{fig:m2}
\end{figure}

\section{Conclusion}

\label{s:c}

We studied an updated version of the income fluctuation problem, the ``common
ancestor'' of modern macroeconomic theory (\cite{ljungqvist2012recursive},
p.~3.)  Working in a setting where returns on financial assets, non-financial
income and impatience are all state dependent and fluctuate over time, we
obtained conditions under which the household savings problem has a unique
solution that can be computed by successive approximations and the wealth
process under the optimal savings policy has a unique stationary distribution
with Pareto right tail.
We also obtained conditions under which wealth is ergodic and exhibits
geometric mixing and asymptotic normality.  We investigated the nature of our
conditions and provided methods for testing them in applications.
While our work was motivated by the desire to better understand the joint
distribution of income and wealth, the income fluctuation problem also has 
applications in asset pricing,
life-cycle choice, fiscal policy, monetary policy, optimal taxation, and
social security.  The ideas contained in this paper should be helpful for
those fields after suitable modifications or extensions.

\appendix

\section{Preliminaries}
\label{s:prelm}

Given a topological space $\TT$, let $\bB(\TT)$ be the Borel
$\sigma$-algebra and $\pP(\TT)$ be the probability measures on
$\bB(\TT)$.
A \emph{stochastic kernel} $\Pi$ on $\TT$ is a map $\Pi \colon \TT \times
\bB(\TT)
\to [0, 1]$ such that $x \mapsto \Pi(x, B)$ is $\bB(\TT)$-measurable for each $B
\in \bB(\TT)$ and $B \mapsto \Pi(x, B)$ is a probability measure on $\bB(\TT)$
for each $x \in \TT$.  For all $t \in \NN$, $x, y \in \TT$ and $B \in
\bB(\TT)$, we define
	$\Pi^1 := \Pi$ and $\Pi^t (x, B) := \int \Pi^{t-1}(y, B) \Pi (x, \diff
    y)$.
Furthermore, for all $\mu \in \pP(\TT)$, let
  $(\mu \Pi^t)(B) := \int \Pi^t (x, B) \mu(\diff x)$.
$\Pi$ is called \emph{Feller} if $x \mapsto \int h(y) \Pi(x, \diff y)$ is
continuous on 
$\TT$ whenever $h$ is bounded and continuous on $\TT$.  We call $\psi \in \pP(\TT)$ \emph{stationary}
for $\Pi$ if $\psi \Pi = \psi$. 

A sequence $\{ \mu_n \} \subset \pP(\TT)$ is called \emph{tight}, if, for all 
$\epsilon>0$, there exists a compact $K \subset \TT$ such that 
$\mu_n(\TT \backslash K) \leq \epsilon$ for all $n$. 
A stochastic kernel $\Pi$ is called \emph{bounded in probability} if the
sequence $\{ Q^t(x, \cdot)\}_{t \geq 0}$ is tight for all $x \in \TT$. 
Given $\mu \in \pP(\TT)$, we define the total variation norm 
$\|\mu\|_{TV} := \sup_{g: |g| \leq 1} \left| \int g \diff \mu \right|$.
Given any measurable map $V \colon \TT \to [1, \infty)$, we say that 
$\Pi$ is \textit{$V$-geometrically mixing} if there exist constants $M < \infty$
and $\lambda < 1$ such that, for all $x \in \TT$ and $t \geq 0$, the corresponding Markov 
process $\{X_t\}$ satisfies $\sup_{k \geq 0; \, h^2, \, g^2 \leq V} \left|
        \EE_x g(X_t) h(X_{t + k}) - \left[\EE_x g(X_t) \right] \left[\EE_x h(X_{t+k}) \right]  
    \right|
    \leq \lambda^t M V(x)$.

Below we use $(\Omega, \fF, \PP)$ to denote a fixed probability space on which
all random variables are defined. $\EE$ is expectations with respect to $\PP$.
The state process $\{Z_t\}$ and the innovation processes $\{\epsilon_t\}$,
$\{\zeta_t\}$ and $\{\eta_t\}$ introduced in \eqref{eq:RY_func} live on this
space.  In what follows, $\{Z_t\}$ is a stationary version of the chain, where
$Z_0$ is drawn from its unique stationary distribution---henceforth denoted
$\pi_Z$.  The marginal distributions of the innovations are denoted by
$\pi_\epsilon$, $\pi_\zeta$ and $\pi_\eta$ respectively.  We let $\{\fF_t\}$
be the natural filtration generated by $\{Z_t\}$ and the
three innovation processes.   $\PP_z$ conditions on $Z_0 = z$ and $\EE_z$ is
expectation under $\PP_z$.

We first prove Lemma~\ref{l:besr}, since its implications will be used
immediately below.  In the proof, we consider the matrix $L_\phi$ as a linear
operator on $\RR^{\ZZ}$ and identify vectors in $\RR^{\ZZ}$ with real-valued
functions on $\ZZ$.

\begin{proof}[Proof of Lemma~\ref{l:besr}]
    A proof by induction confirms that, for any function $h \in \RR^{\ZZ}$, 
    \begin{equation}
        \label{eq:lbev}
        L_\phi^n \, h(z) 
        = \EE_z \prod_{t=1}^n \phi_t h(Z_t),
    \end{equation}
    where $L_\phi^n$ is the $n$-th composition of the operator $L_\phi$ with
    itself (or, equivalently, the $n$-th power of the matrix $L_\phi$).  The
    positivity of $L_\phi$ and Theorem~9.1 of
    \cite{krasnosel2012approximate} imply that $r(L_\phi)
    = \lim_{n \to \infty} \| L_\phi^n \, h \|^{1/n}$ when $\| \cdot \|$ is any
    norm on $\RR^{\ZZ}$ and $h$ is everywhere positive on $\ZZ$.  With $h
    \equiv 1$ and $\| f \| =  \EE |f(Z_0)|$, this becomes
    \begin{equation}
        \label{eq:lbr}
        r(L_\phi) 
        = \lim_{n \to \infty} 
        \left( \EE \, L_\phi^n \, \1(Z_0) \right)^{1/n}
        = \lim_{n \to \infty} 
        \left(\EE\,  \EE_{Z_0} \prod_{t=1}^n \phi_t \right)^{1/n}
        = \lim_{n \to \infty}
        \left(\EE\, \prod_{t=1}^n \phi_t \right)^{1/n}
    \end{equation}
    where the second equality is due to \eqref{eq:lbev} and $h = \1$ and
    the third is by the law of iterated expectations.  
\end{proof}

\begin{lemma}
    \label{l:theta}
    Let $\{\phi_t \}$ and $G_\phi$ be as defined in Lemma~\ref{l:besr}.  If
    $G_\phi < 1$, then there exists an $N$ in $\NN$ and a $\delta < 1$ such
    that $\max_{z \in \ZZ} \EE_z \prod_{t=1}^n \phi_t < \delta^n$ whenever $n \geq N$.
\end{lemma}

\begin{proof}
    Recalling from the proof of Lemma~\ref{l:besr} that 
    $r(L_\phi)
    = \lim_{n \to \infty} \| L_\phi^n \, h \|^{1/n}$ when $\| \cdot \|$ is any
    norm on $\RR^{\ZZ}$ and $h$ is everywhere positive on $\ZZ$, we can again
    take $h \equiv 1$ but now switch to $\| f \| =  \max_{z \in \ZZ} |f(z)|$,
    so that \eqref{eq:lbr} becomes
    \begin{equation}
        \label{eq:lbr2}
        r(L_\phi) 
        = \lim_{n \to \infty} 
        \left( \max_{z \in \ZZ} \, L_\phi^n \, \1(z) \right)^{1/n}
        = \lim_{n \to \infty} 
        \left(\max_{z \in \ZZ} \EE_z \prod_{t=1}^n \phi_t \right)^{1/n}.
    \end{equation}
    Since $r(L_\phi) = G_\phi$ and $G_\phi < 1$, the claim in
    Lemma~\ref{l:theta} now follows.
\end{proof}

\section{Proof of Section \ref{s:ifp} Results}
\label{s:proof_opt}

\begin{proof}[Proof of Proposition~\ref{p:necg}]
    Pick any $a \geq 0$ and $z \in \ZZ$.
    Since $c_t = Y_t$ for all $t$ is dominated by a feasible consumption path, monotonicity
    of $u$ and the law of iterated expectations give
    \begin{align*}
        \max \, \EE_{a,z}
                \sum_{t = 0}^\infty 
                \prod_{i=0}^t \beta_i  u(c_t)
         \geq 
        \EE_z 
                \sum_{t = 0}^\infty 
                \prod_{i=0}^t \beta_i  u(Y_t)
         =
                \sum_{t = 0}^\infty 
                \EE_z \prod_{i=0}^t \beta_i  h(Z_t) ,
    \end{align*}
    where $h(Z_t) := \EE_{Z_t} u(Y)$ and the monotone convergence theorem
    has been employed to pass the expectation through the sum.  In view of
    \eqref{eq:lbev} and $\beta_0 = 1$, we then have
    \begin{equation}
        \label{eq:mgms}
        \max \, \EE_{a,z} 
                \sum_{t = 0}^\infty 
                \prod_{i=0}^t \beta_i  u(c_t)
         \geq 
         \sum_{t = 0}^\infty L_\beta^t \, h(z).
    \end{equation}
    By the assumed almost sure positivity of $\beta_t$ and the irreducibility
    of $P$, the matrix $L_\beta$ is irreducible.  Hence, by the
    Perron--Frobenius theorem, we can choose an everywhere positive
    eigenfunction $e$ such that $L_\beta e = r(L_\beta) e$.  By the everywhere
    positivity of $u(Y_t)$, the function $h$ is everywhere positive on $\ZZ$,
    and hence we can choose $\alpha > 0$ such that $e_\alpha := \alpha e$ is
    less than $h$ pointwise on $\ZZ$. We then have
    \begin{equation*}
        \sum_{t = 0}^\infty L_\beta^t \, h(z)
        \geq \sum_{t = 0}^\infty L_\beta^t \, e_\alpha(z)
        = \alpha \sum_{t = 0}^\infty r(L_\beta)^t \, e(z).
    \end{equation*}
    By lemma~\ref{l:besr} we know that $r(L_\beta) \geq 1$, and since $\alpha$
    and $e$ are positive, this expression is infinite.  Returning to
    \eqref{eq:mgms}, we see that the value function is infinite at our
    arbitrarily chosen pair $(a, z)$.
\end{proof}

For the rest of this section we suppose that Assumptions~\ref{a:b0}--\ref{a:y0} hold.

\begin{lemma}
    \label{l:tefi}
        $M_1 := 
            \sum_{t = 0}^{\infty} \max_{z \in \ZZ} 
            \EE_z  \prod_{i=1}^{t} \beta_i$ and $M_2 := \sum_{t = 0}^{\infty} 
            \max_{z \in \ZZ} \EE_z  \prod_{i=1}^{t} \beta_i R_i$,
        are finite, as are the constants $M_3 = \max_{z \in \ZZ} \EE_z Y$ and $M_4 = \max_{z \in \ZZ} \EE_z u'(Y)$.
\end{lemma}

\begin{proof}
    That $M_1$ and $M_2$ are finite follows directly from Lemma~\ref{l:theta},
    with $\phi_t = \beta_t$ and $\phi_t = \beta_t R_t$ respectively.
    Regarding $M_3$, Assumption~\ref{a:y0} states that $\EE Y < \infty$.
    By the Law of Iterated Expectations, we can write this as $\sum_{z \in
    \ZZ} \EE_z Y \pi_Z(z) < \infty$.  As $\{Z_t\}$ is irreducible, we know
    that $\pi_Z$ is positive everywhere on $\ZZ$.  Hence, $M_3 < \infty$ must
    hold. The proof of $M_4 < \infty$ is similar.
\end{proof}

\begin{lemma}
    \label{lm:max_path}
    For the maximal asset path $\{ \tilde{a}_t \}$ defined by 
    \begin{equation}
        \label{eq:max_path}
        \tilde{a}_{t+1} = R_{t+1} \, \tilde{a}_t + Y_{t+1}
        \quad \text{and}
        \quad (\tilde{a}_0, \tilde{z}_0) = (a,z) \; \text{given},
    \end{equation}
    we have, for each $(a, z) \in \SS_0$, that
        $M(a, z) := \sum_{t = 0}^\infty \EE_{a,z} 
        \prod_{i=0}^t \beta_i \, \tilde{a}_t < \infty$.
\end{lemma}

\begin{proof}
Iterating backward on \eqref{eq:max_path}, we can show 
that
    $\tilde{a}_t = 
        \prod_{i=1}^t R_i \, a + \sum_{j=1}^t Y_j \prod_{i=j+1}^t R_i$.
Taking expectation yields
\begin{equation*}
    \EE_{a,z} \prod_{i=0}^t \beta_i \, \tilde{a}_t 
    =
      \EE_{z} \prod_{i=1}^t \beta_i R_i \, a + 
      \sum_{j=1}^t       
          \EE_{z}             
              \prod_{i=j+1}^t \beta_i R_i
              \prod_{k=0}^j \beta_k \, Y_j.   
\end{equation*}
Then the Monotone Convergence Theorem and the Markov property imply that
\begin{align*}
  M(a,z) &=
  \sum_{t=0}^{\infty} \EE_z \prod_{i=1}^t \beta_i R_i \, a +
  \sum_{t=0}^{\infty} \sum_{j=1}^t       
  \EE_{z} \prod_{i=j+1}^t \beta_i R_i 
      \prod_{k=0}^j \beta_k \, Y_j    \\
  &= \EE_z \sum_{t=0}^{\infty} \prod_{i=1}^t \beta_i R_i \, a +
  \sum_{j=1}^{\infty} \sum_{i=0}^\infty \EE_z 
      \prod_{k=0}^j \beta_k \, Y_j
      \prod_{\ell=1}^{i} \beta_{j+\ell} R_{j+\ell}    \\
  &= \sum_{t=0}^{\infty} \EE_z \prod_{i=1}^t \beta_i R_i \, a +
  \sum_{j=1}^{\infty} \EE_z \prod_{k=0}^j \beta_k \, Y_j \,
      \EE_{Z_j} \sum_{i=0}^\infty \prod_{\ell=1}^{i} \beta_{\ell} R_{\ell}.
\end{align*}
By Lemma~\ref{l:tefi}, we now have, for all $(a,z) \in \SS_0$,
\begin{align*}
    M(a,z) 
    \leq M_2 \, a + M_2 \sum_{t=1}^{\infty} \EE_z \prod_{i=0}^{t} \beta_i Y_t 
    = M_2 \, a + M_2 \sum_{t=1}^{\infty} \EE_z \prod_{i=0}^{t} \beta_i \, \EE_{Z_t} Y.
\end{align*}
Applying Lemma~\ref{l:tefi} again gives $M(a, z)  < \infty$, as was to be shown.
\end{proof}

\begin{proposition}
    \label{pr:Vc}
    The value $V_c(a,z)$ in \eqref{eq:Vc} is well-defined in $\{ -\infty \} \cup \RR$.
\end{proposition}

\begin{proof}
By the assumptions on the utility function, there exists a constant $B \in \RR_+$ such that 
$u(c) \leq c + B$, and hence
    $V_c(a,z) 
    \leq \EE_{a,z} \sum_{t = 0}^\infty \prod_{i=0}^t \beta_i \, u(\tilde{a}_t)    
    \leq M(a, z) + B \sum_{t = 0}^\infty \EE_z \prod_{i=0}^t \beta_i$.
The last term is finite by Lemma~\ref{l:theta}.
\end{proof}

\begin{proof}[Proof of Thoerem~\ref{t:opt_result}]
    The proof is a long but relatively straightforward extension of Theorem~1 of
    \cite{benhabib2015wealth} and thus omitted. A full proof is available
    from the authors upon request.
\end{proof}

\begin{proposition}
    \label{pr:complete}
    $(\cC, \rho)$ is a complete metric space.
\end{proposition}

%\textcolor{red}{Perhaps we can just cite Li and Stachurski here?}

\begin{proof}%[Proof of Proposition~\ref{pr:complete}]
    The proof is a straightforward extension of Proposition~4.1 of \cite{li2014solving} and thus omitted. 
    A full proof is available from the authors upon request.
\end{proof}

%\begin{proof}[Proof of Proposition~\ref{pr:complete}]
%Standard argument shows that $\rho$ is a valid metric. To show completeness 
%of $(\cC, \rho)$, define $\hH$ to be the set of functions $h \colon \SS_0 \to \RR$ that 
%satisfies
%%
%\begin{enumerate}
%    \item $h$ is continuous,
%    \item $h$ is decreasing in the first argument, and
%    \item $\exists M \in \RR$ such that $u'(a) \leq h(a,z) \leq u'(a) + M$
%        for all $(a,z) \in \SS_0$.
%\end{enumerate}
%%
%On $\hH$ we impose the distance
%%
%\begin{equation}
%\label{eq:dinf_metric}
%    d_{\infty}(h,g) 
%      := \left\| h - g \right\|
%      := \sup_{(a,z) \in \SS_0} \left| h(a,z) - g(a,z) \right|.
%\end{equation}
%%
%While the elements of $\hH$ are not bounded, the function $d_{\infty}$ is a 
%valid metric. Moreover, standard argument shows that $(\hH, d_{\infty})$ is 
%a complete metric space. Let us show that $(\cC, \rho)$ and 
%$(\hH, d_{\infty})$ are isometrically isomorphic.
%
%To see that this is so, let $H$ be the map on $\cC$ defined by 
%$Hc = u' \circ c$. It is easy to show that 
%$H: \cC \to \hH$ and that it is a bijection. Moreover, for all 
%$c,d \in \cC$, 
%%
%\begin{equation*}
%    d_{\infty}(Hc, Hd) = \left\| Hc - Hd \right\| 
%                       = \left\| u' \circ c - u' \circ d \right\|
%                       = \rho(c,d).
%\end{equation*}
%%
%Hence, $H$ is an isometry. The space $(\cC, \rho)$ is then complete, as 
%claimed.
%\end{proof}

\begin{proof}[Proof of Proposition~\ref{pr:suff_optpol}]
    Let $c$ be a policy in $\cC$ satisfying \eqref{eq:foc}. To show that any asset
    path generated by $c$ satisfies the transversality condition \eqref{eq:tvc},
    observe that, by condition \eqref{eq:C4}, we have
    \begin{equation}
    \label{eq:bd_uprime}
      c \in \cC \Longrightarrow 
        \exists M \in \RR_+ 
        \text{ s.t. } 
        u'(a) \leq (u' \circ c)(a,z) \leq u'(a) + M, \,
        \forall (a,z) \in \SS_0.
    \end{equation}
    \begin{equation}
    \label{eq:ineq_betu'ca}
    \therefore \quad 
    \EE_{a,z} \prod_{i=0}^t \beta_i \, (u' \circ c) (a_t, Z_t) a_t
    \leq \EE_{a, z} \prod_{i=0}^t \beta_i \, u'(a_t) a_t + 
    M \, \EE_{a, z} \prod_{i=0}^t \beta_i \, a_t .
    \end{equation}
    Regarding the first term on the right hand side of \eqref{eq:ineq_betu'ca}, 
    fix $A > 0$ and observe that
    \begin{align*}
        u'(a_t) a_t 
        &= u'(a_t) a_t \1 \{a_t \leq A\} + u'(a_t) a_t \1 \{a_t > A\}    \\
        &\leq A u'(a_t) + u'(A) a_t   
        \leq A u'(Y_t) + u'(A) \tilde{a}_t
    \end{align*}
    with probability one, where $\tilde{a}_t$  is the maximal path defined in \eqref{eq:max_path}. We 
    then have
    \begin{equation}
    \label{eq:ineq_betEu'a}
        \EE_{a,z} \prod_{i=0}^t \beta_i \, u'(a_t) a_t 
        \leq A \EE_{z} \prod_{i=0}^t \beta_i \, u'(Y_t) 
            + u'(A) \EE_{a,z} \prod_{i=0}^t \beta_i \, \tilde{a}_t.
    \end{equation} 
    By Lemma~\ref{l:tefi}, we have
    \begin{equation*}
            A \, \EE_{z} \prod_{i=0}^t \beta_i \, u'(Y_t) 
            = A \, \EE_{z} \prod_{i=0}^t \beta_i \, \EE_{Z_t} u'(Y) 
            \leq M_4 A \, \EE_{z} \prod_{i=0}^t \beta_i,
    \end{equation*}
    and the last expression converges to zero as $t \to \infty$ by Lemma~\ref{l:theta}.  The second term
    in \eqref{eq:ineq_betEu'a} also converges to zero by Lemma~\ref{lm:max_path}.
    Hence $\EE_{a,z} \prod_{i=0}^t \beta_i \, u'(a_t) a_t  \to 0$ as
    $t \to \infty$, which, combined with \eqref{eq:ineq_betu'ca} and another
    application of Lemma~\ref{lm:max_path}, gives our desired result.
\end{proof}

\begin{proposition}
    \label{pr:welldef_T}
    For all $c \in \cC$ and $(a,z) \in \SS_0$, there exists a unique $\xi \in
    (0,a]$ that solves \eqref{eq:T_opr}.
\end{proposition}

\begin{proof}%[Proof of Proposition~\ref{pr:welldef_T}]
Fix $c \in \cC$ and $(a,z) \in \SS_0$. Because $c \in \cC$, the map 
$\xi \mapsto \psi_c(\xi, a, z)$ is increasing. Since $\xi \mapsto u'(\xi)$
is strictly decreasing, the equation \eqref{eq:T_opr} can have at most one 
solution. Hence uniqueness holds.

Existence follows from the intermediate value theorem provided we can show 
that 
\begin{enumerate}
  \item[(a)] $\xi \mapsto \psi_c(\xi, a, z)$ is a continuous function,
  \item[(b)] $\exists \xi \in (0,a]$ such that 
      $u'(\xi) \geq \psi_c(\xi, a, z)$, and
  \item[(c)] $\exists \xi \in (0,a]$ such that 
      $u'(\xi) \leq \psi_c(\xi, a, z)$.
\end{enumerate}
For part (a), it suffices to show that 
\begin{equation*}
    g(\xi) := 
    \EE_{z} \hat{\beta} \hat{R} 
    \left(u' \circ c \right) 
    [\hat{R}(a - \xi) + \hat{Y}, \hat{Z}] 
\end{equation*}
is continuous on $(0,a]$. To this end, fix $\xi \in (0,a]$ and 
$\xi_n \to \xi$. By \eqref{eq:bd_uprime} we have
\begin{align}
\label{eq:uppbd_ruprmc}
    \hat{\beta} \hat{R} \left( u' \circ c \right) 
        [\hat{R} \left( a - \xi \right) + \hat{Y}, \hat{Z}]  
    \leq 
        \hat{\beta} \hat{R} \left( u' \circ c \right) ( \hat{Y}, \hat{Z} )
    \leq 
        \hat{\beta} \hat{R} u'( \hat{Y}) + \hat{\beta} \hat{R} M.
\end{align}
The last term is integrable, as follows easily from Lemma~\ref{l:tefi}.
Hence the dominated
convergence theorem applies. From this fact and the continuity of $c$, we
obtain $g(\xi_n) \to g(\xi)$. Hence, $\xi \mapsto \psi_c(\xi, a, z)$ is
continuous.

Part (b) clearly holds, since $u'(\xi) \to \infty$ as 
$\xi \to 0$ and $\xi \mapsto \psi_c(\xi, a, z)$ is increasing and 
always finite (since it is continuous as shown in the previous paragraph). 
Part (c) is also trivial (just set $\xi = a$).
\end{proof}

\begin{proposition}
    \label{pr:self_map}
    We have $Tc \in \cC$ for all $c \in \cC$.
\end{proposition}

\begin{proof}
Fix $c \in \cC$ and let
$g \left( \xi, a, z \right) := 
  \EE_{z} \hat{\beta} \hat{R} 
          \left( u' \circ c \right) 
          [\hat{R} \left( a - \xi \right) + \hat{Y}, 
                 \, \hat{Z}]$.

\textbf{Step~1.} We show that $Tc$ is continuous. To apply a standard 
fixed point parametric continuity result such as Theorem~B.1.4 of 
\cite{stachurski2009economic}, we first show that $\psi_c$ is jointly 
continuous on the set $G$ defined in \eqref{eq:dom_T_opr}. This will be true
if $g$ is jointly continuous on $G$. For any $\{ (\xi_n, a_n, z_n) \}$ and 
$(\xi, a, z)$ in $G$ with $(\xi_n, a_n, z_n) \to (\xi, a, z)$, we 
need to show that $g(\xi_n, a_n, z_n) \to g(\xi, a, z)$. To that 
end, we define
\begin{align*}
 h_1 ( \xi, a, \hat{Z}, \hat{\epsilon}, \hat{\zeta}, \hat{\eta} ), \,
  h_2 ( \xi, a, \hat{Z}, \hat{\epsilon}, \hat{\zeta}, \hat{\eta} )   
  := \hat{\beta} \hat{R} 
     [ u' (\hat{Y}) + M ]
       \pm 
     \hat{\beta} \hat{R} \left( u' \circ c \right)
     [ \hat{R} \left( a - \xi \right) + \hat{Y}, \hat{Z} ],
\end{align*}
where $\hat{\beta} := \beta (\hat{Z}, \hat{\epsilon})$, 
$\hat{R} := R ( \hat{Z}, \hat{\zeta} )$ and 
$\hat{Y} := Y ( \hat{Z}, \hat{\eta} )$ as defined in 
\eqref{eq:RY_func}. Then $h_1$ and $h_2$ are continuous in $(\xi, a, \hat{Z})$ by 
the continuity of $c$  and nonnegative by \eqref{eq:uppbd_ruprmc}.

By Fatou's lemma and Theorem~1.1 of \cite{feinberg2014fatou},
\begin{align*}
    & \iiint 
    \sum_{\hat z \in \ZZ}
        h_i ( \xi, a, \hat{z}, \hat{\epsilon}, \hat{\zeta}, \hat{\eta} )
        P( z, \hat{z}) 
    \pi_\epsilon (\diff \hat{\epsilon}) 
    \pi_\zeta(\diff \hat{\zeta}) 
    \pi_\eta(\diff \hat{\eta})    \\
    & \leq 
    \iiint 
    \liminf_{n \to \infty}
    \sum_{\hat z \in \ZZ}
          h_i ( \xi_n, a_n, \hat{z}, \hat{\epsilon}, \hat{\zeta}, \hat{\eta} )
      P( z_n, \hat{z})
      \pi_\epsilon (\diff \hat{\epsilon}) 
      \pi_\zeta(\diff \hat{\zeta}) 
      \pi_\eta(\diff \hat{\eta})    \\
    & \leq \liminf_{n \to \infty}
        \iiint 
        \sum_{\hat z \in \ZZ}
        h_i ( \xi_n, a_n, \hat{z}, \hat{\epsilon}, \hat{\zeta}, \hat{\eta} )
        P( z_n, \hat z)
        \pi_\epsilon (\diff \hat{\epsilon}) 
        \pi_\zeta(\diff \hat{\zeta}) 
        \pi_\eta(\diff \hat{\eta}).
\end{align*}
This implies that
\begin{align*}
 \liminf_{n \to \infty} 
     \left(
      \pm 
      \EE_{z_n}        
         \hat{\beta} \hat{R} 
         \left( u' \circ c \right)
         [\hat{R} \left(a_n - \xi_n \right) + \hat{Y}, \hat{Z}]   
     \right)
 \geq \left( 
      \pm 
      \EE_{z}   
        \hat{\beta} \hat{R} 
        \left(u' \circ c \right)
        [ \hat{R} \left(a - \xi \right) + \hat{Y}, \hat{Z} ]  
      \right).
\end{align*}
The function $g$ is then continuous, since the above inequality is equivalent
to the statement
$\liminf_{n \to \infty} g(\xi_n, a_n, z_n)   
  \geq g(\xi, a, z)   
  \geq   
  \limsup_{n \to \infty} g(\xi_n, a_n, z_n)$.
Hence, $\psi_c$ is continuous on $G$, as was to be shown.
Moreover, since $\xi \mapsto \psi_c(\xi, a, z)$ takes values in the closed 
interval $I(a,z) := [u'(a), u'(a) + \EE_z \hat{\beta} \hat{R} (u'(\hat{Y}) + M)]$,
and the correspondence $(a, z) \mapsto I(a,z)$ is nonempty, compact-valued and 
continuous, Theorem~B.1.4 of \cite{stachurski2009economic} then implies that $Tc$ 
is continuous on $\SS_0$.

\textbf{Step 2.} We show that $Tc$ is increasing in $a$. Suppose that for 
some $z \in \ZZ$ and $a_1, a_2 \in (0, \infty)$ with $a_1 < a_2$, we have
$\xi_1 := Tc (a_1,z) > Tc (a_2,z) =: \xi_2$. Since $c$ is increasing in $a$ 
by assumption, $\psi_c$ is increasing in $\xi$ and decreasing in $a$. Then 
$u'(\xi_1) < u'(\xi_2) 
 = \psi_c(\xi_2, a_2, z) 
 \leq \psi_c(\xi_1, a_1, z) = u'(\xi_1)$. This is a contradiction.
 
\textbf{Step 3.} We have shown in Proposition~\ref{pr:welldef_T} that 
$Tc(a,z) \in (0,a]$ for all $(a,z) \in \SS_0$.

\textbf{Step 4.} We show that $\| u' \circ (Tc) - u' \| < \infty$. Since
$u'[Tc(a,z)] \geq u'(a)$, we have
\begin{align*}
    &\left| u'[Tc(a,z)] - u'(a) \right| 
    = u'[Tc(a,z)] - u'(a)    \\
    & \leq
      \EE_{z} \hat{\beta} \hat{R} 
              \left( u' \circ c \right)  
              ( \hat{R} [ a - Tc(a,z) ] + \hat{Y}, \, \hat{Z})    
    \leq \EE_{z} \hat{\beta} \hat{R} [ u'(\hat{Y}) + M ]
\end{align*}
for all $(a,z) \in \SS_0$. The right hand side is easily shown to be finite
via Lemma~\ref{l:tefi}.
\end{proof}

To prove Theorem~\ref{t:ctra_T}, let $\hH$ be all continuous functions $h:\SS_0 \to \RR$ that is decreasing in its first argument and $(a,z) \mapsto h(a,z)-u'(a)$ is bounded and nonnegative.
%Recall $\hH$ defined in the proof of Proposition~\ref{pr:complete}. 
Given $h \in \hH$, let $\tilde{T} h$ be the function mapping $(a,z) \in \SS_0$ into the $\kappa$ that solves
\begin{equation}
\label{eq:T_hat}
\kappa = \max \{ \EE_{z} \, \hat{\beta} \hat{R}  \,
    h (\hat{R} \, [a - (u')^{-1}(\kappa)] + \hat{Y}, \hat{Z}), \, u'(a) \}.
\end{equation}
Moreover, consider the bijection $H: \cC \to \hH$ defined by $Hc := u' \circ c$.
%The next lemma implies that $\tilde{T}$ is a well-defined self-map on $\hH$, 
%as well as topologically conjugate to $T$ under the bijection 
%$H: \cC \to \hH$ defined by $Hc := u' \circ c$.

\begin{lemma}
\label{lm:conjug}
  The operator $\tilde{T} \colon \hH \to \hH$ and satisfies $\tilde{T} H = H T $ on $\cC$.
\end{lemma}

\begin{proof}%[Proof of Lemma~\ref{lm:conjug}]
Pick any $c \in \cC$ and $(a,z) \in \SS_0$. Let $\xi := Tc(a,z)$, then $\xi$ solves
\begin{equation}
\label{eq:Tc_eq}
    u'(\xi) = 
     \max \{ \EE_{z} \, \hat{\beta} \hat{R} 
           \left(u' \circ c \right) [\hat{R} \left(a - \xi \right) + \hat{Y}, \hat{Z}],
           \, u'(a) \}.
\end{equation}
We need to show that $HTc$ and $\tilde{T} Hc$ evaluate to the same number at 
$(a,z)$. In other words, we need to show that $u'(\xi)$ is the solution to
\begin{equation*}
  \kappa = 
    \max \{ \EE_{z} \, \hat{\beta} \hat{R}
          \left( u' \circ c \right)
          (\hat{R} \, [a - (u')^{-1} (\kappa)] + \hat{Y}, \hat{Z}),
          \, u'(a) \}.
\end{equation*}
But this is immediate from \eqref{eq:Tc_eq}. Hence, we have shown that
$\tilde{T} H = H T$ on $\cC$. Since $H \colon \cC \to \hH$ is a bijection,
we have $\tilde{T} = HT H^{-1}$. Since in addition $T \colon \cC \to \cC$ by 
Proposition~\ref{pr:self_map}, we have $\tilde{T} \colon \hH \to \hH$. This 
concludes the proof.
\end{proof}

\begin{lemma}
    \label{lm:monot}
    $\tilde{T}$ is order preserving on $\hH$. That is, 
    $\tilde{T} h_1 \leq \tilde{T} h_2$ for all $h_1, h_2 \in \hH$ with $h_1 \leq h_2$.
\end{lemma}

\begin{proof}%[Proof of Lemma~\ref{lm:monot}]
Let $h_1, h_2$ be functions in $\hH$ with $h_1 \leq h_2$. Suppose to the 
contrary that there exists $(a,z) \in \SS_0$ such that 
$\kappa_1 := \tilde{T} h_1 (a,z) > \tilde{T} h_2 (a,z) =: \kappa_2$. Since functions in $\hH$ are decreasing in the first argument, we have
\begin{align*}
  \kappa_1 &= 
      \max \{ \EE_{z} \, \hat{\beta} \hat{R} \, 
              h_1 ( \hat{R} \, [a - (u')^{-1}(\kappa_1)] + \hat{Y}, \hat{Z}),
              \, u'(a) \}      \\
     & \leq 
      \max \{ \EE_{z} \, \hat{\beta} \hat{R} \,
              h_2 ( \hat{R} \, [a - (u')^{-1}(\kappa_1)] + \hat{Y}, \hat{Z}),
              \, u'(a) \}      \\
     & \leq 
      \max \{ \EE_{z} \, \hat{\beta} \hat{R} \, 
              h_2 ( \hat{R} \, [a - (u')^{-1}(\kappa_2)] + \hat{Y}, \hat{Z}),
              \, u'(a) \}
     = \kappa_2.
\end{align*}
This is a contradiction. Hence, $\tilde{T}$ is order preserving.
\end{proof}

\begin{lemma}
    \label{lm:ctra_That}
    There exists an $n \in \NN$ and $\theta < 1$ such that 
    $\tilde{T}^n$ is a contraction mapping of modulus $\theta$ on $(\hH, d_{\infty})$.
\end{lemma}

\begin{proof}%[Proof of Lemma~\ref{lm:ctra_That}]
    Since $\tilde{T}$ is order preserving and $\hH$ is closed under the
    addition of nonnegative constants, based on
    \cite{blackwell1965discounted}, it remains to verify the existence of $n
    \in \NN$ and $\theta < 1$ such that
      $\tilde{T}^n (h+\gamma) \leq \tilde{T}^n h + \theta \gamma$
        for all $h \in \hH $  and $\gamma \geq 0$.
    By Lemma~\ref{l:theta} and Assumption~\ref{a:rb0}, it suffices to show that for all $k \in \NN$ and
    $(a,z) \in \SS_0$, we have
    \begin{equation}
    \label{eq:That_k}
        \tilde{T}^k (h+ \gamma) (a,z)
        \leq \tilde{T}^k h(a,z) + \gamma \, \EE_z \prod_{i=1}^k \beta_i R_i .
    \end{equation}
    Fix $h \in \hH$, $\gamma \geq 0$, and let $h_{\gamma} (a,z) := h(a, z) +
    \gamma$. By the definition of $\tilde{T}$, we have
    \begin{align*}
    \tilde{T} h_\gamma (a,z) 
    &= \max \{ \EE_{z} \, \hat{\beta} \hat{R} \, 
              h_{\gamma} (\hat{R} \, [a - (u')^{-1}(\tilde{T} h_{\gamma})(a,z)] + \hat{Y}, \hat{Z}),
              u'(a) \}    \\
    &\leq 
       \max \{ \EE_{z} \, \hat{\beta} \hat{R} \, 
               h ( \hat{R} \, [a - (u')^{-1} (\tilde{T} h_{\gamma})(a,z)] + \hat{Y}, \hat{Z}), u'(a) \}
        + \gamma \EE_z \beta_1 R_1    \\
    &\leq
       \max \{ \EE_{z} \, \hat{\beta} \hat{R} \,
               h ( \hat{R} \, [a - (u')^{-1} (\tilde{T} h)(a,z)] + \hat{Y}, \hat{Z}), u'(a) \}
        + \gamma \EE_z \beta_1 R_1.
    \end{align*}
    Here, the first inequality is elementary and the second is due to the fact
    that $h \leq h_\gamma$ and $\tilde{T}$ is order preserving. Hence,
    $\tilde{T} (h+ \gamma) (a,z ) \leq \tilde{T} h (a,z) + \gamma \EE_z
    \beta_1 R_1$ and \eqref{eq:That_k} holds for $k = 1$. Suppose
    \eqref{eq:That_k} holds for arbitrary $k$. It remains to show that it
    holds for $k+1$. For $z \in \ZZ$, define
        $f(z) := \gamma \EE_z \beta_1 R_1 \cdots \beta_k R_k$.
    By the induction hypothesis, the monotonicity of $\tilde{T}$ and the Markov property,
    \begin{align*}
        \tilde{T}^{k+1} h_\gamma (a,z)
        &= \max \{ \EE_{z} \, \hat{\beta} \hat{R} \, (\tilde{T}^k h_{\gamma})
            (\hat{R} \, [a - (u')^{-1}(\tilde{T}^{k+1} h_{\gamma})(a,z)] + \hat{Y}, \hat{Z}), u'(a) \}    \\
       & \leq \max \{ \EE_z \, \hat{\beta} \hat{R} \, (\tilde{T}^k h + f) 
           (\hat{R} \, [a - (u')^{-1}(\tilde{T}^{k+1} h_{\gamma})(a,z)] + \hat{Y}, \hat{Z}), u'(a) \}    \\
       & \leq \max \{ \EE_z \, \hat{\beta} \hat{R} \, (\tilde{T}^k h) (\hat{R} \, 
           [a - (u')^{-1}(\tilde{T}^{k+1} h_{\gamma})(a,z)] + \hat{Y}, \hat{Z}), u'(a) \}    \\
       & \quad + \EE_{z} \, \beta_1 R_1 f(Z_1)    \\
       & \leq  \max \{ \EE_z \, \hat{\beta} \hat{R} \, (\tilde{T}^k h) ( \hat{R} \,
           [a - (u')^{-1}(\tilde{T}^{k+1} h)(a,z)] + \hat{Y}, \hat{Z}), u'(a) \}     \\ 
       &\quad + \gamma \EE_{z} \, \beta_1 R_1 \, \EE_{Z_1} \beta_1 R_1 \cdots \beta_k R_k    \\
       &= \tilde{T}^{k+1} h(a,z) + \gamma \EE_z \, \beta_1 R_1 \cdots \beta_{k+1} R_{k+1}.
    \end{align*} 
    Hence, \eqref{eq:That_k} is verified by induction. This concludes the proof.
\end{proof}

\begin{proof}[Proof of Theorem~\ref{t:ctra_T}]
    Let $n$ and $\theta$ be as in Lemma~\ref{lm:ctra_That}.
    In view of Propositions \ref{pr:suff_optpol}, \ref{pr:complete} and
    \ref{pr:self_map}, to show that $T^n$ is a contraction and verify claims
    (1)--(3) of Theorem~\ref{t:ctra_T}, based on the Banach contraction
    mapping theorem, it suffices to show that
      $\rho(T^n c, T^n d) \leq \theta \rho(c,d)$ for all $c,d \in \cC$.
    To this end, pick any $c,d \in \cC$. Note that the topological conjugacy result established in Lemma~\ref{lm:conjug} implies that $\tilde{T} = H T H^{-1}$. Hence,
        $\tilde{T}^n = (H T H^{-1}) \cdots (H T H^{-1}) = H T^n H^{-1} $ and 
        $\tilde{T}^n H = H T^n$.
    By the definition of $\rho$ and the contraction property established in Lemma~\ref{lm:ctra_That},
    \begin{equation*}
      \rho(T^n c, T^n d) = d_{\infty}(H T^n c, H T^n d) 
                   = d_{\infty}(\tilde{T}^n Hc, \tilde{T}^n Hd)
                   \leq \theta d_{\infty}(Hc, Hd)
                   = \theta \rho(c,d).
    \end{equation*}
    Hence, $T^n$ is a contraction and claims (1)--(3) are verified.
\end{proof}

Our next goal is to prove Proposition~\ref{pr:monotonea}. To begin with, we define
\begin{equation*}
\cC_0 = \left\{
c \in \cC \colon a \mapsto a - c(a,z) \text{ is increasing for all } z \in \ZZ \right\}.
\end{equation*}

\begin{lemma}
	\label{lm:cC'}
	$\cC_0$ is a closed subset of $\cC$, and $Tc \in \cC_0$ for all $c \in \cC_0$. 
\end{lemma}

\begin{proof}%[Proof of Lemma~\ref{lm:cC'}]
	To see that $\cC_0$ is closed, for a given sequence $\{c_n\}$ in $\cC_0$ and $c \in \cC$ with $\rho(c_n, c) \to 0$, we need to show that $c \in \cC_0$. This obviously holds since $a \mapsto a - c_n(a,z)$ is increasing for all $n$, and, in addition, $\rho(c_n,c) \to 0$ implies that $c_n (a,z) \to c(a,z)$ for all $(a,z) \in \SS_0$.
	
	Fix $c \in \cC_0$. We now show that $\xi := Tc \in \cC_0$. Since $\xi \in
	\cC$ by Proposition~\ref{pr:self_map}, it remains to show that $a \mapsto
	a - \xi(a,z)$ is increasing. Suppose the claim is false, then there exist
	$z \in \ZZ$ and $a_1, a_2 \in (0, \infty)$ such that 
	$a_1 < a_2 $
	and $a_1 - \xi(a_1, z) > a_2 - \xi (a_2, z)$.
	Since $a_1 - \xi (a_1, z) \geq 0$, $a_2 - \xi (a_2, z) \geq 0$ and $\xi(a_1,z) \leq \xi (a_2, z)$ by Proposition~\ref{pr:self_map}, we have
	$\xi(a_1, z) < a_1$ and  $\xi(a_1, z) < \xi(a_2, z)$.
	However, based on the property of the time iteration operator, we then have
	\begin{align*}
	(u' \circ \xi) (a_1, z) 
	&= \EE_z \hat{\beta} \hat{R} (u'\circ c) 
	( \hat{R} \, [a_1 - \xi(a_1,z)] + \hat{Y}, \hat{Z} )    \\
	&\leq \EE_z \hat{\beta} \hat{R} (u'\circ c) 
	( \hat{R} \, [a_2 - \xi(a_2,z)] + \hat{Y}, \hat{Z} ) 
	\leq (u' \circ \xi) (a_2, z),
	\end{align*}
	which implies that $\xi(a_1, z) \geq \xi(a_2, z)$. This is a contradiction. Hence, 
	$a \mapsto a - \xi (a,z)$ is increasing, and $T$ is a self-map on $\cC_0$.
\end{proof}

\begin{proof}[Proof of Proposition~\ref{pr:monotonea}]
	%Since $T$ is a self-map on the closed subset $\cC'$ by Lemma~\ref{lm:cC'} and $T c^* = c^*$, we have $c^* \in \cC'$ by Theorem~\ref{t:ctra_T}. Hence, the stated claims hold.
	Since $T$ maps elements of the closed subset $\cC_0$ into itself by Lemma~\ref{lm:cC'}, Theorem~\ref{t:ctra_T} implies that $c^* \in \cC_0$. Hence, the stated claims hold.
\end{proof}

\begin{proof}[Proof of Proposition~\ref{pr:monotoneY}]
	Let $T_j$ be the time iteration operator for the income process $j$
	established in Proposition~\ref{pr:self_map}. It suffices to show $T_1c
	\leq T_2c$ for all $c \in \cC$. To see this, note that by
	Lemma~\ref{lm:monot}, we have $T_jc_1 \leq T_jc_2$ whenever $c_1 \leq
	c_2$. Therefore if $T_1c \leq T_2c$ for all $c \in \cC$, we obtain $T_1c_1
	\leq T_1c_2 \leq T_2c_2$. Iterating this starting from any $c\in \cC$, by
	Theorem~\ref{t:ctra_T}, it follows that
	$c_1^* = \lim_{n \to \infty}(T_1)^nc \leq \lim_{n \to \infty}(T_2)^nc
	= c_2^*$,
	completing the proof.
	
	To show that $T_1c \leq T_2c$ for any $c\in \cC$, take any $(a,z)\in
	\SS_0$ and define $\xi_j=(T_jc)(a,z)$. To show $\xi_1 \leq \xi_2$, suppose
	on the contrary that $\xi_1 > \xi_2$. Since $c$ is increasing in $a$ and
	$u''<0$ (hence $u'$ is decreasing), it follows from the definition of the
	time iteration operator in \eqref{eq:T_opr}--\eqref{eq:keypart_T_opr},
	$Y_1 \leq Y_2$, $u''<0$ and the monotonicity of $c \in \cC$ that
	\begin{align*}
	u'(\xi_2)>u'(\xi_1)
	&=\max \{ \EE_{z} \, \hat{\beta} \hat{R} \, (u' \circ c) [\hat{R}(a - \xi_1) + \hat{Y}_1, 
	\hat{Z}], u'(a) \}    \\
	&\geq \max \{ \EE_{z} \, \hat{\beta} \hat{R} \, (u' \circ c)[\hat{R}(a - \xi_2) + \hat{Y}_2, 
	\hat{Z}], u'(a) \} 
	= u'(\xi_2),
	\end{align*}
	which is a contradiction.
\end{proof}

To prove Proposition~\ref{pr:optpol_concave}, we need several lemmas.

\begin{lemma}
	\label{lm:binding}
	For all $c \in \cC_0$, there exists a threshold $\bar{a}_c(z)$ such that $Tc(a,z) = a$ if and only if $a \leq \bar{a}_c (z)$. In particular, there exists a threshold $\bar{a}(z)$ such that $c^*(a,z) = a$ if and only if $a \leq \bar{a}(z)$.
\end{lemma}

\begin{proof}%[Proof of Lemma~\ref{lm:binding}]
	Recall that, for all $c \in \cC_0$, $\xi(a,z) := Tc(a,z)$ solves
	\begin{equation}
	\label{eq:T_opr_general}
	\left( u' \circ \xi \right)(a,z) =
	\max \{ \EE_{z} \, \hat{\beta} \hat{R} \left(u' \circ c \right) 
	(\hat{R} \, [a - \xi(a,z)] + \hat{Y}, \hat{Z}), u'(a) \}.
	\end{equation}
	%
	%Let $c^* \in \cC$ denote the optimal policy. 
	For each $z \in \ZZ$ and $c \in \cC_0$, define
	\begin{equation}
	\label{eq:a_bar}
	\bar{a}_c (z) := 
	\left(u' \right)^{-1}
	[\EE_z \, \hat{\beta} \hat{R} 
	\left(u' \circ c \right) ( \hat{Y}, \hat{Z} ) ]    
	\quad \text{and} \quad
	\bar{a}(z) := \bar{a}_{c^*} (z).
	\end{equation}
	To prove the first claim, by Lemma~\ref{lm:cC'}, it suffices to show that $\xi(a,z) < a$ implies $a > \bar{a}_c(z)$. This  obviously holds since in view of \eqref{eq:T_opr_general}, the former implies that
	\begin{align*}
	u'(a) < \EE_z \, \hat{\beta} \hat{R} 
	\left(u' \circ c \right) 
	(\hat{R} \left[a - \xi(a,z) \right] + \hat{Y}, \hat{Z})    
	\leq \EE_z \, \hat{\beta} \hat{R} 
	\left(u' \circ c \right) (\hat{Y}, \hat{Z})
	= u'[\bar{a}_c (z)],
	\end{align*}
	which then yields $a > \bar{a}_c(z)$. The second claim follows immediately from the first claim and the fact that $c^* \in \cC_0$ is the unique fixed point of  $T$ in $\cC$.
\end{proof}

Consider a subset $\cC_1$ defined by
$\cC_1 := \left\{ c \in \cC_0 \colon 
a \mapsto c(a,z) \text{ is concave for all }
z \in \ZZ 
\right\}$.

\begin{lemma}
	\label{lm:self_map_cC1}
	$\cC_1$ is a closed subset of $\cC_0$ and $\cC$, and, $T c \in \cC_1$ for all $c \in \cC_1$.
\end{lemma}

\begin{proof}%[Proof of Lemma~\ref{lm:self_map_cC1}]
	The first claim is immediate because limits of concave functions are
	concave. To prove the second claim, fix $c \in \cC_1$. We have $Tc \in
	\cC_0$ by Lemma~\ref{lm:cC'}. It remains to show that $a \mapsto
	\xi(a, z) := Tc (a,z)$ is concave for all $z \in \ZZ$. Given $z \in \ZZ$,
	Lemma~\ref{lm:binding} implies that $\xi(a,z) = a$ for $a \leq
	\bar{a}_c(z)$ and that $\xi (a,z) < a$ for $a > \bar{a}_c(z)$.  Since in
	addition $a \mapsto \xi(a,z)$ is continuous and increasing, to show the
	concavity of $\xi$ with respect to $a$, it suffices to show that $a
	\mapsto \xi (a,z)$ is concave on $(\bar{a}_c(z), \infty)$.
	
	Suppose there exist some $z \in \ZZ$, 
	$\alpha \in [0,1]$, and $a_1, a_2 \in (\bar{a}_c (z), \infty)$ such that 
	\begin{equation}
	\label{eq:as_ctdt}
	\xi \left( (1-\alpha) a_1 + \alpha a_2, \, z \right)
	< (1 - \alpha) \xi(a_1, z) + \alpha \xi(a_2, z).
	\end{equation}
	Let $h(a, z, \hat{\omega}):= \hat{R} \left[a - \xi(a, z) \right] +
	\hat{Y}$, where $\hat{\omega} := (\hat{R}, \hat{Y})$. Then by
	Lemma~\ref{lm:binding} and noting that consumption is interior, we have
	\begin{align*}
	(u' \circ \xi) \left( (1 - \alpha) a_1 + \alpha a_2, \, z \right)  
	&= \EE_z \, \hat{\beta} \hat{R} \left(u' \circ c \right) 
	\{ h [(1 - \alpha) a_1 + \alpha a_2, \, z, \, \hat{\omega}], \hat{Z}\}    \\
	&\leq \EE_z \, \hat{\beta} \hat{R} \left( u' \circ c \right) 
	[(1-\alpha) h(a_1, z, \hat{\omega}) + \alpha h(a_2, z, \hat{\omega}), \hat{Z}].
	\end{align*}
	Using condition~\eqref{eq:concave_prop} then yields 
	\begin{align*}
	&\xi ((1 - \alpha) a_1 + \alpha a_2, z)  
	\geq (u')^{-1} \{ \EE_z \, \hat{\beta} \hat{R} (u' \circ c) 
	[(1 - \alpha) h(a_1, z, \hat{\omega}) + \alpha h(a_2, z, \hat{\omega}), \hat{Z} ] \}    \\
	& \geq (1 -\alpha) (u')^{-1} 
	\{ \EE_z \, \hat{\beta} \hat{R} (u' \circ c) 
	[h(a_1, z, \hat{\omega}), \hat{Z}] \} +
	\alpha (u')^{-1} \{ \EE_z \, \hat{\beta} \hat{R}
	(u' \circ c) [h(a_2, z, \hat{\omega}), \hat{Z}] \}    \\ 
	& = (1-\alpha) (u')^{-1} \{ (u' \circ \xi) (a_1, z) \} +
	\alpha (u')^{-1} \{ (u' \circ \xi) (a_2, z) \}   
	= (1 - \alpha) \xi(a_1, z) + \alpha \xi(a_2, z),
	\end{align*}
	which contradicts \eqref{eq:as_ctdt}. Hence, $a \mapsto \xi(a,z)$ is concave for all $z \in \ZZ$.
\end{proof}

\begin{proof}[Proof of Proposition~\ref{pr:optpol_concave}]
	By Theorem~\ref{t:ctra_T}, $T \colon \cC \to \cC$ is a 
	contraction mapping with unique fixed point $c^*$. Since $\cC_1$ is a closed
	subset of $\cC$ and $T \cC_1 \subset \cC_1$ by Lemma~\ref{lm:self_map_cC1}, 
	we know that $c^* \in \cC_1$. The first claim is verified.
	Regarding the second claim, note that $c^* \in \cC_1$ implies that 
	$a \mapsto c^*(a,z)$ is increasing and concave for all $z \in \ZZ$. Hence, 
	$a \mapsto c^*(a,z)/a$ is a decreasing function for all $z \in \ZZ$. 
	Since $0\leq c^*(a,z) \leq a$ for all $(a,z) \in \SS_0$, $\alpha(z) := \lim_{a \to \infty} c^*(a,z)/a$ is well-defined and $\alpha(z) \in [0,1]$.
\end{proof}

\begin{proof}[Proof of Remark \ref{rm:concave}]
	For each $c$ in $\cC$ concave in its first argument, let 
	$h_c (x, \hat{\omega}) 
	:= c ( \hat{R} x + \hat{Y}, \hat{z} )$, 
	where $\hat{\omega} := ( \hat{R}, \hat{Y}, \hat{z} )$. 
	Then $x \mapsto h_c (x, \hat{\omega})$ is concave. Based on the generalized Minkowski's inequality 
	(see, e.g., \cite{hardy1952inequalities}, page~146, theorem 198), we have 
	\begin{align*}
	&[ \EE_z \, \hat{\beta} \hat{R} \, 
	h_c(\alpha x_1 + (1 - \alpha) x_2, \hat{\omega})^{-\gamma} ]^{-\frac{1}{\gamma}}   
	\geq \{ \EE_z \, \hat{\beta} \hat{R} 
	\left[
	\alpha h_c (x_1, \hat{\omega}) +
	(1 - \alpha) h_c (x_2, \hat{\omega}) 
	\right]^{-\gamma} 
	\}^{-\frac{1}{\gamma}}   \\
	&= \{ \EE_z [ \, \alpha (\hat{\beta} \hat{R})^{-\frac{1}{\gamma}} h_c(x_1, \hat{\omega})
	+ (1 - \alpha) (\hat{\beta} \hat{R})^{-\frac{1}{\gamma}} h_c(x_2, \hat{\omega})
	\,]^{-\gamma}
	\}^{-\frac{1}{\gamma}}    \\
	&\geq ( \EE_z [ \alpha (\hat{\beta} \hat{R})^{-\frac{1}{\gamma}} 
	h_c(x_1, \hat{\omega}) ]^{-\gamma} )^{-\frac{1}{\gamma}}  + 
	( \EE_z [ (1 - \alpha) (\hat{\beta} \hat{R})^{-\frac{1}{\gamma}} h_c(x_2, \hat{\omega})
	]^{-\gamma} 
	)^{-\frac{1}{\gamma}}   \\
	&= \alpha [ \EE_z \, \hat{\beta} \hat{R} \, h_c(x_1, \hat{\omega})^{-\gamma} 
	]^{-\frac{1}{\gamma}}  + 
	(1 - \alpha) [ \EE_z \, \hat{\beta} \hat{R} \, h_c(x_2, \hat{\omega})^{-\gamma} 
	]^{-\frac{1}{\gamma}},
	\end{align*}
	Since $u'(c) = c^{-\gamma}$, the above inequality implies that condition \eqref{eq:concave_prop} holds.
\end{proof}

To prove Proposition~\ref{pr:optpol_linbound}, let $\bar s$ be as in \eqref{eq:suff_linbound} and define 
\begin{equation}
\label{eq:cC2}
\cC_2 := \left\{ 
c \in \cC: c(a, z) \geq (1 - \bar s) a 
\; \text{ for all } (a,z) \in \SS_0
\right\}.
\end{equation}

\begin{lemma}
	\label{lm:cC2}
	$\cC_2$ is a closed subset of $\cC$, and $T c \in \cC_2$ for all $c \in \cC_2$.
\end{lemma}

\begin{proof}%[Proof of Lemma~\ref{lm:cC2}]
	To see that $\cC_2$ is closed, for a given sequence $\{ c_n \}$ in $\cC_2$ and
	$c \in \cC$ with $\rho( c_n , c) \to 0$, we need to verify that $c \in \cC_2$.
	This obviously holds since $c_n(a,z) /a \geq 1 - \bar s$ for all $n$ and $(a,z)
	\in \SS_0$, and, on the other hand, $\rho( c_n , c) \to 0$ implies that $c_n
	(a,z) \to c(a,z)$ for all $(a,z) \in \SS_0$. 
	
	We next show that $T$ is a self-map on $\cC_2$. Fix $c \in \cC_2$. We have $Tc \in \cC$ since $T$ is a self-map on $\cC$. It remains to show that $\xi := Tc$ satisfies 
	$\xi (a,z) \geq (1 - \bar s) a$ for all $(a,z) \in \SS_0$. Suppose  
	$\xi(a,z) < (1 - \bar s) a$ for some $(a,z) \in \SS_0$. Then 
	\begin{equation*}
	u'((1 - \bar s) a) < (u' \circ \xi)(a, z)    
	= \max \{ \EE_z \, \hat{\beta} \hat{R} \left( u' \circ c \right) 
	(\hat{R} \left[ a - \xi(a,z) \right] + \hat{Y}, \hat{Z}), \, u'(a) \}. 
	\end{equation*}
	Since $u'((1 - \bar s) a) > u'(a)$ and $c \in \cC_2$, this implies that
	\begin{align*}
	u'((1 - \bar s) a )
	&< \EE_z \, \hat{\beta} \hat{R}    
	\left( u' \circ c \right) 
	( \hat{R} \left[ a - \xi(a,z) \right] + \hat{Y}, \hat{Z} )     \\
	&\leq \EE_z \, \hat{\beta} \hat{R}    
	\, u' \, \{ 
	(1 - \bar s) \hat{R} \left[ a - \xi(a,z) \right] 
	+ (1 - \bar s) \hat{Y} \}    \\
	&\leq \EE_z \, \hat{\beta} \hat{R}    
	\, u' \, [
	(1 - \bar s) \hat{R} \bar s a + (1 - \bar s) \hat{Y} ]    
	\leq \EE_z \, \hat{\beta} \hat{R}  \,  
	u' \, [ \hat{R} \bar s (1 - \bar s) a],
	\end{align*}
	which contradicts \eqref{eq:suff_linbound} since $((1 - \bar s) a, z) \in \SS_0$.
	As a result, $\xi(a,z) \geq (1 - \bar s) a$ for all $(a,z) \in \SS_0$ and we conclude
	that $Tc \in \cC_2$.
\end{proof}

\begin{proof}[Proof of Proposition~\ref{pr:optpol_linbound}]
	We have shown in Theorem~\ref{t:ctra_T} that $T$ is a contraction mapping
	on the complete metric space $(\cC, \rho)$, with unique fixed point $c^*$.
	Since in addition $\cC_2$ is a closed subset of $\cC$ and $T \cC_2 \subset
	\cC_2$ by Lemma~\ref{lm:cC2}, we know that $c^* \in \cC_2$. The stated
	claim is verified.
\end{proof}

\section{Proof of Section \ref{s:sto_stability} Results}
\label{s:proof_stab}

As before, Assumptions~\ref{a:b0}--\ref{a:y0} are in force.
%To prove Theorem~\ref{t:sta_exist}, 
Notice that Assumption~\ref{a:rb0}, Assumption~\ref{a:r0} and
Lemma~\ref{l:theta} imply existence of an $n$ in $\NN$ such that 
\begin{equation}
    \label{eq:geo_coeff}
    \theta := \max_{z \in \ZZ} \EE_z \prod_{t=1}^n \beta_t R_t < 1
    \quad \, \text{and} \, \quad
    \gamma := \bar s^n \max_{z \in \ZZ} \EE_z \prod_{t=1}^n R_t < 1.
%    \quad \text{and} \quad
%    M_0 := \sup_{1 \leq \ell \leq n, \, z \in \ZZ} \EE_z \prod_{t=1}^\ell R_t < \infty.
\end{equation}

\begin{lemma}
	\label{lm:bd_in_prob_at}
	For all $(a,z) \in \SS$, we have $\sup_{t \geq 0} \EE_{a,z} \, a_t  < \infty$.
\end{lemma}

\begin{proof}%[Proof of Lemma~\ref{lm:bd_in_prob_at}]
    Since $c^*(0,z) = 0$, Proposition~\ref{pr:optpol_linbound} implies that
    $c^*(a,z) \geq (1 - \bar s) a$ for all $(a,z) \in \SS$. For all $t \geq 1$, we
    have $t = kn + j$ in general, where the integers $k \geq 0$ and $j \in
    \{0,1, \dots, n-1\}$. Using these facts and \eqref{eq:trans_at}, we have:
	\begin{align*}
	a_t & \leq \bar s^t R_t \cdots R_1 a + \bar s^{t-1} R_t \cdots R_2 Y_1
	    + \cdots + \bar s R_t Y_{t-1} + Y_t    \\
	    &= \bar s^{kn+j} R_{kn +j} \cdots R_1 a + \sum_{\ell=1}^{j} \bar s^{kn + j - \ell} R_{kn+j} \cdots R_{\ell+1} Y_{\ell}   
        \\
	    & \quad + \sum_{m=1}^{k} \sum_{\ell=1}^{n} \bar s^{mn - \ell} R_{kn + j} \cdots R_{(k-m)n + j + \ell + 1} Y_{(k-m)n + j + \ell}
	\end{align*}
	with probability one.  Taking expectations of the above while noting that
    $M_0 := \max_{1 \leq \ell \leq n, \, z \in \ZZ} \EE_z \prod_{t=1}^\ell R_t < \infty$ by
    Assumption~\ref{a:r0} and Lemma \ref{l:theta}, we have 
	\begin{align*}
		\EE_{a,z} a_t 
		&\leq \gamma^k \bar s^{j} \EE_z R_{j} \cdots R_1 a + 
		\gamma^k \sum_{\ell=1}^{j} \bar s^{j - \ell} \EE_z R_{j} \cdots R_{\ell+1} Y_{\ell}    \\
		& \quad + \sum_{m=0}^{k-1} \gamma^m \sum_{\ell=1}^{n} \bar s^{n - \ell}
		\EE_z R_{(k-m)n + j} \cdots R_{(k-m-1)n + j + \ell + 1} Y_{(k-m)n + j + \ell}     \\
		&\leq \gamma^k M_0 a + \gamma^k M_0 \sum_{\ell=1}^{j} \EE_z Y_\ell +
		\sum_{m=0}^{k-1} \gamma^m M_0 \sum_{\ell=1}^{n} \EE_z Y_{(k-m-1)n + j + \ell}  \\
		&\leq M_0 a + M_0 M_3 n + \sum_{m=0}^{\infty} \gamma^m M_0 M_3 n < \infty.
	\end{align*}
    or all $(a,z) \in \SS$ and $t \geq 0$.  Here we have used $M_3$ in
    Lemma~\ref{l:tefi} and the Markov property.  Hence, $\sup_{t \geq 0}
    \EE_{a,z} \, a_t < \infty$ for all $(a,z) \in \SS$, as was claimed.
\end{proof}

A function $w^* \colon \SS \to \RR_+$ is called \emph{norm-like} if all 
its sublevel sets (i.e., sets of the form 
$\{x \in \SS \colon w(x) \leq b \}, b \in \RR_+$) are precompact in $\SS$ 
(i.e., any sequence in a given sublevel set has a subsequence that converges 
to a point of $\SS$).

\begin{proof}[Proof of Theorem~\ref{t:sta_exist}]
    Based on Lemma~D.5.3 of \cite{meyn2009markov}, a stochastic kernel $Q$ is
    bounded in probability if and only if for all $x \in \SS$, there exists a
    norm-like function $w_x^* \colon \SS \to \RR_+$ such that the
    $(Q,x)$-Markov process $\{X_t\}_{t \geq 0}$ satisfies $\limsup_{t \to
    \infty} \EE_x \left[ w_x^*(X_t) \right] < \infty$. 
	Fix $(a,z) \in \SS$. Since $\ZZ$ is finite, $P$ is bounded in probability. 
	Hence, there exists a norm-like function 
	$w \colon \ZZ \to \RR_+$ such that 
	$\limsup_{t \to \infty} \EE_z w (Z_t) < \infty$.
	Then $w^* \colon \SS \to \RR_+$ defined by 
	$w^*(a_0,Z_0) := a_0 + w (Z_0)$ is 
	a norm-like function on $\SS$. The stochastic kernel $Q$ is then bounded in 
	probability since Lemma~\ref{lm:bd_in_prob_at} implies that
	  $\limsup_{t \to \infty} \EE_{a,z} \, w^*(a_t, Z_t)
	  \leq \sup_{t \geq 0} \EE_{a,z} \, a_t +
	       \limsup_{t \to \infty} \EE_z \, w(Z_t) 
	  < \infty$.
	Regarding existence of stationary distribution, since $P$ is Feller 
	(due to the finiteness of $\ZZ$), whenever
	$z_n \to z$, the product measure satisfies
	\begin{equation*}
	  P(z_n, \cdot) \otimes \pi_\zeta \otimes \pi_\eta 
	  \stackrel{w}{\longrightarrow}
	  P(z, \cdot) \otimes \pi_\zeta \otimes \pi_\eta.
	\end{equation*}
	Since in addition $c^*$ is continuous, a simple application of the 
	generalized Fatou's lemma of \cite{feinberg2014fatou} (Theorem~1.1) 
	shows that the stochastic kernel $Q$ is Feller. Moreover, since $Q$ 
	is bounded in probability, based on the Krylov-Bogolubov theorem (see, e.g., 
	\cite{meyn2009markov}, Proposition~12.1.3 and Lemma~D.5.3), $Q$ admits at 
	least one stationary distribution.
\end{proof}

\begin{lemma}
	\label{lm:bind_fntime}
    The borrowing constraint binds in finite time with positive probability.
    That is, for all $(a,z) \in \SS$, we have
	$\PP_{a,z} \left( \cup_{t \geq 0} \{ c_t = a_t \} \right) > 0$.
\end{lemma}

\begin{proof}%[Proof of Lemma~\ref{lm:bind_fntime}]
	The claim holds trivially when $a=0$. Suppose the claim does not hold on $\SS_0$ (recall that $\SS_0 = \SS \backslash \{0\}$), then $\PP_{a,z} \left( \cap_{t \geq 0} \{c_t < a_t\} \right) = 1$ for some $(a,z) \in \SS_0$, i.e., the borrowing constraint never binds with probability one. Hence, 
	\begin{equation*}
	    \PP_{a,z} 
	     \left\{ 
	        (u' \circ c^*)(a_t, Z_t) = 
	        \EE \left[ \beta_{t+1} R_{t+1} 
	                     (u' \circ c^*) (a_{t+1}, Z_{t+1}) \big| \fF_{t} 
	                  \right]
	     \right\}
	    = 1
	\end{equation*}
	for all $t \geq 0$. Then we have
	\begin{align}
	\label{eq:u'c_ineq}
	  \left( u' \circ c^* \right)(a,z)
	  &= \EE_{a,z} \, \beta_1 R_1 \cdots \beta_t R_t 
	                \left(u' \circ c^* \right)(a_t, Z_t)    \nonumber  \\
	  & \leq \EE_{a,z} \, \beta_1 R_1 \cdots \beta_t R_t  
	                    \left[ u'(a_t) + M \right]   
	  \leq \EE_{z} \, \beta_1 R_1 \cdots \beta_t R_t
	                    \left[ u'(Y_t) + M \right]
	\end{align}
	for all $t \geq 1$. Let $n$ and $\theta$ be defined by \eqref{eq:geo_coeff}. Let $t= kn +1$. Based on the Markov property and Lemma \ref{l:tefi}, as $k \to \infty$, 
	\begin{align*}
		&\EE_z \beta_1 R_1 \cdots \beta_t R_t    
		%&= \EE_z \beta_1 R_1 \cdots \beta_{t-1} R_{t-1} \EE_z (\beta_t R_t \mid \fF_{t-1})    
		= \EE_z \beta_1 R_1 \cdots \beta_{t-1} R_{t-1} \EE_{Z_{t-1}} \beta_1 R_1    \\
		&\leq \left(\max_{z \in \ZZ} \EE_z \beta_1 R_1 \right) 
		    (\EE_z \beta_1 R_1 \cdots \beta_{nk} R_{nk})
		    \leq \left(\max_{z \in \ZZ} \EE_z \beta_1 R_1 \right) \theta^k
		\to 0.
	\end{align*}
	Similarly, as $k \to \infty$, 
	\begin{align*}
		&\EE_z \, \beta_1 R_1 \cdots \beta_t R_t u'(Y_t)
	%	&= \EE_z \, \beta_1 R_1 \cdots \beta_{t-1} R_{t-1} 
	%	          \EE_z \left[ \beta_t R_t u'(Y_t) \mid \fF_{t-1} \right]     \\
		= \EE_z \, \beta_1 R_1 \cdots \beta_{t-1} R_{t-1} \EE_{Z_{t-1}} 
		    \left[\beta_1 R_1 u'(Y_1) \right]     \\
		&\leq \left[ \max_{z \in \ZZ} \EE_z \hat{\beta} \hat{R} u'(\hat{Y}) \right]
		    \EE_z \beta_1 R_1 \cdots \beta_{nk} R_{nk}    
		\leq \left[ \max_{z \in \ZZ} \EE_z \hat{\beta} \hat{R} u'(\hat{Y}) \right] \theta^k
		\to 0.
	\end{align*}
	Letting $k \to \infty$. \eqref{eq:u'c_ineq} then implies that
	$\left(u' \circ c^* \right)(a,z) \leq 0$, contradicted with the fact that $u'>0$. 
	Thus, we must have $\PP_{a,z} \left( \cup_{t \geq 0} \{ c_t = a_t\} \right) > 0$ 
	for all $(a,z) \in \SS$. 
\end{proof}

Our next goal is to prove Theorem~\ref{t:gs_gnl_ergo_LLN}. In proofs we apply the theory of \cite{meyn2009markov}. Important definitions  (their information in the textbook) include: $\psi$-irreducibility (Section~4.2), small set (page~102), strong aperiodicity (page~114), 
petite set (page~117), Harris chain (page~199), and positivity (page~230). 

Recall that $\RR^m$ paired with its Euclidean topology is a second countable 
topological space (i.e., its topology has a countable base). Since $\RR_+$ and 
$\ZZ$ are respectively Borel subsets of $\RR$ and $\RR^m$ paired with the 
relative topologies, they are also second countable. Hence, 
$\SS := \RR_+ \times \ZZ$ satisfies $\bB(\SS) = \bB (\RR_+) \otimes \bB(\ZZ)$ 
(see, e.g., page~149, Theorem~4.44 of \cite{guide2006infinite}).
Recall \eqref{eq:yden}. With slight abuse of notation, in proofs, we use $f$ to denote the 
density of $\{Y_t\}$ in both cases (Y1) and (Y2) and write $\diff y = \nu (\diff y)$, where
$\nu$ is the related measure. Specifically, $\nu$ is the Lebesgue measure when (Y2) holds. Moreover, Let 
$\vartheta$ be the counting measure.

Recall $\bar{z} \in \ZZ$ and the greatest lower bound $y_\ell \geq 0$ of the
support of $\{Y_t\}$ given by Assumption~\ref{a:pos_dens}. Let $\bar{p} :=
P(\bar{z}, \bar{z})$. Then $\bar{p}>0$ by Assumption~\ref{a:pos_dens}.

\begin{lemma}
	\label{lm:bind_fntime_2}
		$\PP_{(a, \bar{z})} \left\{
		\cup_{t \geq 0} \left[
		\{c_t = a_t\} \cap \left( \cap_{i=0}^t \{Z_i = \bar{z} \} \right)
		\right]
		\right\} > 0$
	for all $a \in (0, \infty)$.
\end{lemma}

\begin{proof}%[Proof of Lemma~\ref{lm:bind_fntime_2}]
	Fix $a \in (0, \infty)$. If $a \leq \bar{a}(\bar{z})$, the claim holds trivially by 
	Lemma~\ref{lm:binding}. Now consider the case $a > \bar{a}(\bar{z})$. 
	Suppose 
		$\PP_{(a, \bar{z})} \left\{
		\cup_{t \geq 0} \left[
		\{c_t = a_t\} \cap \left( \cap_{i=0}^t \{Z_i = \bar{z} \} \right)
		\right]
		\right\} = 0$.
	Then, based on the De Morgan's law, we have
	\begin{equation*}
		\PP_{(a, \bar{z})} \left\{
		\cap_{t \geq 0} \left[
		\{c_t < a_t\} \cup \left( \cup_{i=0}^t \{Z_i \neq \bar{z} \} \right)
		\right]
		\right\} = 1.
	\end{equation*}
	\begin{equation*}
        \fore
		\PP_{(a, \bar{z})} 
		\left\{
		\{c_t < a_t\} \cup \left( \cup_{i=0}^t \{Z_i \neq \bar{z} \} \right)
		\right\} = 1
        \text{ for all } t \in \NN.
	\end{equation*}
	\begin{equation*}
	    \fore 
	    \PP_{(a, \bar{z})} 
	    \left\{
	    \{c_k < a_k\} \cup \left( \cup_{i=0}^t \{Z_i \neq \bar{z} \} \right)
	    \right\} = 1
	    \text{ for all } k, t \in \NN
	    \text{ with } k \leq t.
	\end{equation*}
	\begin{equation*}
		\fore 
		\PP_{(a, \bar{z})} 
		\left\{
		    \left( \cap_{i=0}^t \{c_i < a_i\} \right) 
		    \cup \left( \cup_{i=0}^t \{Z_i \neq \bar{z} \} \right)
		\right\} = 1
		\text{ for all } t \in \NN.
	\end{equation*}
	Note that the set 
	$\triangle (t) := \left( \cap_{i=0}^t \{c_i < a_i\} \right) 
	\cup \left( \cup_{i=0}^t \{Z_i \neq \bar{z} \} \right)$ 
	can be written as
	\begin{equation*}
		\triangle (t) =  \triangle_1 (t) \cup  \triangle_2 (t),
		\quad \text{ where } \, \triangle_1 (t) \cap  \triangle_2 (t) = \emptyset,
	\end{equation*}
	\begin{equation*}
		\triangle_1 (t) :=  
		\left( \cap_{i=0}^t \{c_i < a_i\} \right) 
		\cap \left( \cap_{i=0}^t \{Z_i = \bar{z} \} \right)
		\; \text{ and } \;
		\triangle_2 (t) := \cup_{i=0}^t \{Z_i \neq \bar{z} \}.  
	\end{equation*}
    Assumption~\ref{a:pos_dens} then implies that, for all $t \geq 0$,
	\begin{align*}
		\PP_{(a,\bar{z})} \{ \triangle_1 (t) \} = 1 - \PP_{\bar{z}} \{ \triangle_2 (t) \} 
		= \PP_{\bar{z}} \left\{ \cap_{i=0}^t \{Z_i = \bar{z} \} \right\} 
		= \bar{p}^t > 0.
	\end{align*}
	Let $n$ and $\theta$ be defined by \eqref{eq:geo_coeff} and let $t = kn + 1$. 
	Similar to the proof of Lemma~\ref{lm:binding}, we can show that, with probability 
	$\bar{p}^t > 0$, 
	\begin{equation*}
		(u' \circ c^*) (a,\bar{z}) \leq 
		\theta^k \left[ 
		\max_{z \in \ZZ} \EE_z \hat{\beta} \hat{R} u' (\hat{Y}) + 
		M \max_{z \in \ZZ} \EE_z \hat{\beta} \hat{R}
		\right]
	\end{equation*}
	for some constant $M \in \RR_+$. Since $\theta \in (0,1)$ and $(u' \circ c^*) (a,\bar{z}) > 0$, 
	Lemma~\ref{l:tefi} implies that there exists $N \in \NN$ such that 
	\begin{equation*}
		\theta^N \left[ 
		\max_{z \in \ZZ} \EE_z \hat{\beta} \hat{R} u' (\hat{Y}) + 
		M \max_{z \in \ZZ} \EE_z \hat{\beta} \hat{R}
		\right]
		< (u' \circ c^*) (a,\bar{z}).
	\end{equation*}
	As a result, we have $(u' \circ c^*)(a,\bar{z}) < (u' \circ c^*) (a,\bar{z})$  with probability
    $\bar{p}^{Nn + 1} > 0$. This is a contradiction. Hence the stated claim is verified.
\end{proof}

Let $F(\diff a_{t+1} \mid a_t, Z_t, Z_{t+1})$ be defined such that 
$\PP \{ a_{t+1} \in A \mid (a_t, Z_t, Z_{t+1}) = (a,z,z') \} 
= \int \1 \{a' \in A \} F(\diff a' \mid a, z, z')$
at  $A \in \bB (\RR_+)$.

\begin{lemma}
	\label{lm:mono_in_a}
	Let $h : \SS \to \RR_+$ be an integrable map such that $a \mapsto h(a,z)$
	is decreasing for all $z \in \ZZ$. Then, for all $t \in \NN$ and $z \in
	\ZZ$, the map $a \mapsto \ell (a,z, t) := 
	\int h(a',z') Q^t ((a,z), \diff (a',z'))$ is decreasing.
\end{lemma}

\begin{proof}%[Proof of Lemma~\ref{lm:mono_in_a}]
	Fix $z \in \ZZ$. When $t = 1$, \eqref{eq:dyn_sys} implies that
	\begin{equation*}
	\ell (a,z,1) = \int \left[ 
	\int h(a',z') F(\diff a' \mid a, z, z')
	\right] 
	P(z, z') \vartheta (\diff z').
	\end{equation*}
	Since $a \mapsto h(a,z)$ is decreasing, and by
	Proposition~\ref{pr:monotonea} and \eqref{eq:dyn_sys}, the optimal asset
	accumulation path $a_{t+1}$ is increasing in $a_t$ with probability one,
	we know that 
	$a \mapsto \int h(a',z') F(\diff a' \mid a, z, z') $
	is decreasing for all $z' \in \ZZ$. 
	Thus, $a \mapsto \ell (a,z, 1)$ is decreasing. The claim holds for $t=1$. Suppose this claim holds for arbitrary $t$, it remains to show that it holds for $t+1$. Note that
	\begin{align*}
	\ell (a,z, t+1) 
	%&= \int h(a',z') Q^{t+1} ((a,z), \diff (a',z'))    \\
	&= \iint h(a'',z'') Q^t ((a',z'), \diff (a'',z'')) Q ((a,z), \diff (a',z'))    \\
	&= \int \ell (a',z', t) Q ((a,z), \diff (a',z')).
	\end{align*}
	Since $a' \mapsto \ell (a',z',t)$ is decreasing for all $z' \in \ZZ$,
	based on the induction argument, $a \mapsto \ell (a, z, t+1)$ is
	decreasing. The stated claim then follows.
\end{proof}

\begin{lemma}
	\label{lm:psi_irr}
	The Markov process $\{ (a_t,Z_t) \}_{t \geq 0}$ is $\psi$-irreducible.
\end{lemma}

\begin{proof}%[Proof of Lemma~\ref{lm:psi_irr}]
	Recall $\delta > y_\ell$ given by Assumption \ref{a:pos_dens}. Let $\DD \in \bB(\SS)$ be defined by
	$\DD := \{ y_\ell \} \times \{\bar{z}\}$ if (Y1) holds and 
	$\DD := (y_\ell, \delta ) \times \{\bar{z}\}$ if (Y2) holds.
	We define the measure $\varphi$ on $\bB(\SS)$ by
	$\varphi(A) := (\nu \times \vartheta) (A \cap \DD)$ for 
	$A \in \bB (\SS)$.
	Clearly $\varphi$ is a nontrivial measure. In particular, $\vartheta
	(\{\bar{z}\}) = 1$ as $\vartheta$ is the counting measure. Moreover, since
	$y_\ell$ is the greatest lower bound of the support of $\{Y_t\}$, it must
	be the case that $\nu(\{y_\ell\}) > 0$ if (Y1) holds and that
	$\nu((y_\ell, \delta)) > 0$ if (Y2) holds. As a result, $\varphi(\SS) =
	\nu (\{y_\ell\}) \times \vartheta (\{\bar{z}\}) > 0$ when (Y1) holds and
	$\varphi(\SS) = \nu ((y_\ell, \delta)) \times \vartheta (\{\bar{z}\}) >0$ when (Y2) holds.
	
	We first show that $\{(a_t, Z_t)\}$ is $\phi$-irreducible. Let $A$ be an element of $\bB(\SS)$ such that $\varphi (A) > 0$. Fix $(a,z) \in \SS$. We need to show that $\{(a_t, Z_t) \}$ visits set $A$ in finite time with positive probability.
	
	Since $\{z_t\}$ is irreducible,  
	$\PP_z \{Z_{N_0} = \bar{z} \} > 0$ for some integer $N_0 \geq 0$. By Lemma~\ref{lm:bd_in_prob_at}, there exists $\tilde{a} < \infty$ such that 
	$\PP_{(a,z)} \{ a_{N_0} < \tilde{a}, Z_{N_0} = \bar z \} > 0$.
	By Lemma~\ref{lm:bind_fntime_2}, there exists $T \in \NN$ such that 
		$\PP_{(\tilde a, \bar{z})} \left\{
			c_{T} = a_{T}, \, 
			Z_{T} = \bar{z} 
		\right\}
		\geq 
		\PP_{(\tilde a, \bar{z})} \left\{
			c_{T} = a_{T}, \, 
			\cap_{i=0}^{T} \{Z_i = \bar{z} \} 
		\right\} > 0$.
	Lemma~\ref{lm:binding} and Lemma~\ref{lm:mono_in_a} then imply that
	$\PP_{(a', \bar z)} \left\{ c_{T} = a_{T}, \, Z_{T} = \bar{z} \right\} > 0$ for all $a' \in (0, \tilde{a})$. Hence, for $N:=N_0 + T$ and $E:= \left\{ c_{N} = a_{N}, \, Z_N = \bar{z} \right\}$, we have
	\begin{equation}
		\label{eq:binding&inC}
		\PP_{(a, z)} (E) 
		\geq \int_{\{a' \leq \tilde a, \, z'= \bar z \}} 
		    \PP_{(a',\bar z)} \{c_T = a_T, Z_T = \bar z \}
	    Q^{N_0} ((a,z), \diff (a',z')) > 0
		%\; \text{ where } \,
		%E := \left\{ c_{N} = a_{N}, \, Z_N = \bar{z} \right\}.
	\end{equation}
	based on the Markov property. By \eqref{eq:dyn_sys}, we have 
	\begin{align}
		\label{eq:bind_ineq}
		\PP_{(a,z)} \{ (a_{N+1}, Z_{N+1}) \in A \}   
		&\geq \PP_{(a,z)} \left\{ 	        (a_{N+1}, Z_{N+1}) \in A, \, a_N = c_N, \, Z_N = \bar{z} \right\}    
            \nonumber \\
		&= \PP_{(a,z)} \left\{ (a_{N+1}, Z_{N+1}) \in A \mid a_N = c_N, \, Z_N = \bar{z} \right\} \, \PP_{(a,z)} (E)    
            \nonumber \\
		& = \PP_{(a,z)} \left\{ 	        (Y_{N+1}, Z_{N+1}) \in A, \, a_N = c_N, \, Z_N = \bar{z} \right\}.
	\end{align}
	Note that, by Assumption~\ref{a:pos_dens}, 
	$f (y'' \mid z'') P(\bar{z}, z'') > 0$ whenever $(y'', z'') \in \DD$. Since in 
	addition $\varphi (A) = (\nu \times \vartheta)(A \cap \DD) > 0$, we have
	\begin{equation*}
	\int_A f (y'' \mid z'') P(\bar{z}, z'') 
	(\nu \times \vartheta) [\diff (y'', z'')] > 0.
	\end{equation*}
	Let $\triangle := \PP_{(a,z)} \{ (a_{N+1}, Z_{N+1}) \in A \}$. Then \eqref{eq:binding&inC} and \eqref{eq:bind_ineq} imply that
	\begin{align*}
	\triangle 
	&\geq \int_E 
	\left\{
	\int_A f (y'' \mid z'') P(z', z'') 
	(\nu \times \vartheta) [\diff (y'', z'')]
	\right\}
	Q^N \left( (a,z), \diff (a', z') \right) > 0.
	\end{align*}
	Therefore, we have shown that any measurable subset with positive $\varphi$ 
	measure can be reached in finite time with positive probability, i.e., 
	$\{ (a_t,Z_t)\}$ is $\varphi$-irreducible. Based on Proposition~4.2.2 of 
	\cite{meyn2009markov}, there exists a maximal probability measure $\psi$ on $\bB(\SS)$ such that 
	$\{ (a_t,Z_t)\}$ is $\psi$-irreducible.
\end{proof}

\begin{lemma}
	\label{lm:inf_abar}
	Let the function $\bar{a}$ be defined as in \eqref{eq:a_bar}. Then
    $\bar{a}(\bar{z}) \geq y_\ell$ if (Y1) holds, while $\bar{a}(\bar{z}) >
    y_\ell$ if (Y2) holds.
\end{lemma}

\begin{proof}%[Proof of Lemma~\ref{lm:inf_abar}]	

	Suppose (Y1) holds and $\bar{a}(\bar{z}) < y_\ell$. Then, by Lemma~\ref{lm:binding}, for all $t \in \NN$, 
	\begin{align}
	\label{eq:subset}
		\left[ 
		\{c_t = a_t\} \cap 
		\left( \cap_{i=0}^t \{Z_i = \bar{z} \} \right) 
		\right] \,
		& = \, \left[
		\{ a_t \leq \bar{a}(Z_t) \} \cap
		\left( \cap_{i=0}^t \{ Z_i = \bar{z} \} \right)
		\right]    \nonumber \\
		& \subset \, \left[ 
		\{ a_t < y_\ell \} \cap
		\left( \cap_{i=0}^t \{ Z_i = \bar{z} \} \right)
		\right]
		\, \subset \, \{a_t < y_\ell \}.
	\end{align}
	%
	%Hence, for all $a \in ( \bar{a}(\bar{z}), \infty )$, $c_0 < a_0$ by Lemma~\ref{lm:binding}, and, for all $t \in \NN$, 
	Hence, for all $a \in (0, \infty)$ and $t \in \NN$,
	\begin{align*}
	\PP_{(a,\bar{z})} \left[ 
	    \{ c_t = a_t \} \cap \left( \cap_{i=0}^t \{ Z_i = \bar{z} \} \right)
	\right]
	\leq \PP_{(a, \bar{z})} \{ a_t < y_\ell \} = 0,
	\end{align*} 
	where the last equality follows from \eqref{eq:dyn_sys}, which implies
    that $a_t \geq Y_t \geq y_\ell$ with probability one. This is contradicted with Lemma~\ref{lm:bind_fntime_2}. 
	
	Suppose (Y2) holds and $\bar{a}(\bar{z}) \leq y_\ell$. By definition, $\PP_z \{Y_t \leq y_\ell \} = 0$ for all $z \in \ZZ$ and $t \in \NN$.
%	%
%	\begin{equation*}
%		\PP_z \{Y_t \leq \underline{y} \} = 0
%		\; \text{ for all } 
%		\, z \in \ZZ \,
%		\text{ and } \, t \in \NN.
%	\end{equation*}
%	%
	Since $a_t \geq Y_t$ with probability one, we have
		$\PP_{(a,z)} \{a_t \leq y_\ell \} = 0$
		 for all  
         $(a,z) \in \SS$ and $t \in \NN$.
	Via similar analysis to \eqref{eq:subset}, Lemma~\ref{lm:binding} implies that  
		$\left[ 
		\{ c_t = a_t \} \cap \left( \cap_{i=0}^t \{ Z_i = \bar{z} \} \right)
		\right]
		\, \subset \, \{ a_t \leq y_\ell \}$ 
    for all $t \in \NN$.
	%Hence, for all $a \in (\bar{a}(\bar{z}), \infty)$, $c_0 < a_0$ by Lemma~\ref{lm:binding}, and, for all $t \in \NN$, 
	Hence, for all $a \in (0,1)$ and $t \in \NN$, we have
		$\PP_{(a,\bar{z})} \left[ 
		\{ c_t = a_t \} \cap \left( \cap_{i=0}^t \{ Z_i = \bar{z} \} \right)
		\right]
		\leq \PP_{(a, \bar{z})} \{ a_t \leq y_\ell \} = 0$.
	Again, this contradicts Lemma~\ref{lm:bind_fntime_2}. 
\end{proof}

\begin{lemma}
	\label{lm:str_aperi}
	The Markov process $\{ (a_t, Z_t)\}_{t \geq 0}$ is strongly aperiodic.
\end{lemma}

\begin{proof}%[Proof of Lemma~\ref{lm:str_aperi}]
	By the definition of strong aperiodicity, we need to show that there exists 
	a $v_1$-small set $\DD_1$ with $v_1 (\DD_1) > 0$, i.e., there exists a nontrivial 
	measure $v_1$ on $\bB (\SS)$ and a subset $\DD_1 \in \bB (\SS)$ such that 
	$v_1 (\DD_1) > 0$ and
	\begin{equation}
	\label{eq:small}
	\inf_{(a,z) \in \DD_1}  
	Q \left((a,z), A \right) \geq v_1 \left( A \right)
	\quad \text{for all }
	A \in \bB (\SS).
	\end{equation}
	For $\delta > 0$ given by Assumption~\ref{a:pos_dens}, let 
	$\CC := \left( y_\ell, \min \left\{ \delta, \, \bar{a}(\bar z) \right\} \right)$ 
	and let 
		$\DD_1 := \{ y_\ell \} \times \{\bar{z} \}$ if (Y1) holds and
		$\DD_1 := \CC \times \{\bar{z}\}$ if (Y2) holds.
	We now show that $\DD_1$ satisfies the above conditions. Define
	    $r(a', z') := f(a' \,|\, z') P(\bar{z}, z')$
	and note that $r(a',z') > 0$ on $\DD_1$. Define the measure $v_1$ on $\bB(\SS)$ by 
		$v_1 (A) := \int_A r(a',z') (\nu \times \vartheta) [\diff (a',z')]$.
	If (Y1) holds, then $\nu(\{y_\ell \}) > 0$ as shown above, 
	and, if (Y2) holds, Lemma~\ref{lm:inf_abar} implies that $\nu(\CC) > 0$. 
	Since in addition $\vartheta(\{\bar{z}\}) > 0$, it always holds that
	$(\nu \times \vartheta) (\DD_1) > 0$.
	Moreover, since $r(a',z') >0$ on $\DD_1$, we have $v_1 (\DD_1) > 0$ 
	and $v_1$ is a nontrivial measure.
	
%	Let $g[(a',z') \mid (a,z)]$ denote the density representation of the stochastic 
%	kernel $Q$ when $(a,z) \in \DD_1$. Lemma~\ref{lm:binding} implies that
%	%
%	\begin{equation*}
%		g[(a',z') \mid (a,z)] = f (a' \mid z') p(z'\mid z),
%		\qquad (a,z) \in \DD_1.
%	\end{equation*}
%	%
	For all $(a,z) \in \DD_1$ and $A \in \bB(S)$, Lemma~\ref{lm:binding} implies that
	\begin{equation*}
		Q \left((a,z), A \right) 
		= \int_A r(a',z') (\nu \times \vartheta) [\diff (a',z')]
		= v_1 (A). 
	\end{equation*}
	Hence, $\DD_1$ satisfies \eqref{eq:small} and $\{(a_t, Z_t)\}_{t \geq 0}$ is strongly aperiodic.
\end{proof}

\begin{lemma}
	\label{lm:petite_set}
	The set $[0, d] \times \ZZ$ is a petite set for all $d \in \RR_+$. 
\end{lemma}

\begin{proof}%[Proof of Lemma~\ref{lm:petite_set}]
	Fix $d \in (0, \infty)$ and $z \in \ZZ$. Let $B:= [0,d] \times \{z\}$. By Lemma~\ref{lm:bind_fntime_2},
	\begin{equation}
		\label{eq:binding_d}
		\PP_{(d, z)} \{ c_{N-1} = a_{N-1}, Z_{N-1} = \bar{z} \} > 0
		\quad \text{for some } \, N \in \NN. 
	\end{equation}
	We start by showing that there exists a nontrivial measure $v_N$
	on $\bB (\SS)$ such that 
	\begin{equation}
	\label{eq:vN_small}
	\inf_{(a,z) \in B} Q^N ((a,z), A) \geq v_N (A) 
	\quad \text{for all }  A \in \bB(\SS).
	\end{equation}
	In other words, $B$ is a $v_N$-small set. Fix $A \in \bB(\SS)$. For all $z' \in \ZZ$, define
	\begin{equation*}
		m(z') := \int \left[
		\int \1 \{(y'', z'') \in A \} f (y'' \mid z'') \diff y''
		\right] 
		P(z', z'') \vartheta (\diff z'').
	\end{equation*}
	Note that for all $(a,z) \in B$, Lemma~\ref{lm:binding} implies that
	\begin{align*}
		Q^N((a,z), A) 
%		&\geq \PP_{a,z} \left\{ 
%		(a_N, Z_N) \in A, \, c_{N-1} = a_{N-1}, \, Z_{N-1} = \bar{z} 
%		\right\}    \\
%		&= \PP_{a,z} \left\{ 
%		(Y_N, Z_N) \in A, \, c_{N-1} = a_{N-1}, \, Z_{N-1} = \bar{z}
%		\right\}    \\
		&\geq \PP_{a,z} \left\{ 
		(Y_N, Z_N) \in A, \, a_{N-1} \leq \bar{a}(Z_{N-1}), \, Z_{N-1} = \bar{z}
		\right\}    \\
		&= \int m(z') \1 \{a' \leq \bar{a}(z'), z' = \bar{z} \}
		Q^{N-1} ((a,z), \diff (a',z')).     
	\end{align*}
	Since $a' \mapsto m(z') \1 \{a' \leq \bar{a}(z'), z' = \bar{z} \}$ 
	is decreasing for all $z' \in \ZZ$, by Lemma~\ref{lm:mono_in_a},
	\begin{align*}
		Q^N ((a,z), A)
		&\geq \int m(z') \1 \{a' \leq \bar{a}(z'), z' = \bar{z} \}
		Q^{N-1} ((d,z), \diff (a',z'))   \\
		&= \PP_{d,z} \left\{ 
		    (Y_N, Z_N) \in A, \, c_{N-1} = a_{N-1}, \, Z_{N-1} = \bar{z}  
		\right\}
		=: v_N (A).
	\end{align*}
    Note that $v_N$ is a nontrivial measure on $\bB(\SS)$ since
    \eqref{eq:binding_d} implies that $v_N (\SS) > 0$. Furthermore, since
    $(a,z)$ is chosen arbitrarily, the above inequality implies that
    \eqref{eq:vN_small} holds. We have shown that $B$ is a $v_N$-small set,
    and hence a petite set.
    Since finite union of petite sets is petite for $\psi$-irreducible chains
    (see, e.g., Proposition~5.5.5 of \cite{meyn2009markov}), the set $[0,d]
    \times \ZZ$ must also be petite.
\end{proof}

%For the next lemma, let $\bar s \in [0,1)$ be defined as in
%Assumption~\ref{a:r0}. Let $n \in \NN$ and $\gamma \in (0,1)$ be defined by
%\eqref{eq:geo_coeff}.  Let $B:= [0,d] \times \ZZ$, where $d \in (0, \infty)$.
Recall $\bar s \in [0,1)$ in Assumption~\ref{a:r0}, $n \in \NN$ and $\gamma \in (0,1)$ in \eqref{eq:geo_coeff}. Let $B:= [0,d] \times \ZZ$.

\begin{lemma}
	\label{lm:geo_drift}
	There exist constants $b \in \RR_+$, $\rho \in (0,1)$ and a measurable map 
	$V \colon \SS \to [n / \rho, \infty)$ that is bounded on $B$, such that,
	for sufficiently large $d \in \RR_+$ and all $(a,z) \in \SS$, we have
		$\EE_{a,z} V(a_n, Z_n) - V(a,z) 
		\leq - \rho V(a,z) + b \1 \{(a,z) \in B \}$.
\end{lemma}

\begin{proof}%[Proof of Lemma~\ref{lm:geo_drift}]
	Since $c^*(a,z) \geq (1 - \bar s) a$ by 
	Proposition~\ref{pr:optpol_linbound} and $M_0 := \max_{z \in \ZZ} \EE_z \hat{R} < \infty$ by 
	Assumption~\ref{a:r0} and Lemma \ref{l:theta}, by Lemma~\ref{l:tefi} and the Markov property, 
	\begin{align*}
		\EE_{a,z} a_n 
		&\leq \bar s^n \EE_z R_n \cdots R_1 a + 
		\sum_{t = 1}^{n} \bar s^{n-t} \EE_z R_n \cdots R_{t+1} Y_t    \\
		& \leq \gamma a + \sum_{t=1}^{n} \bar s^{n-t} \EE_z Y_t \,  \EE_{Z_t} R_{t+1} \cdots R_n
		\leq \gamma a + \sum_{t=1}^{n} \bar s^{n-t} M_0^{n-t} M_3.
	\end{align*}
    Define $b_0 := \sum_{t=1}^{n} {\bar s}^{n-t} M_0^{n-t} M_3$. Note that $b_0 < \infty$.
    Choose $\rho \in (0, 1 - \gamma)$, $m_V \geq n / \rho$ and $d \in \RR_+$ such that 
    $(1 - \gamma - \rho) d \geq b_0 + \rho m_V$. Then, for $V(a,z) := a + m_V$,
	\begin{align}
		\label{eq:dri_ineq1}
		\EE_{a,z} V(a_n, Z_n) - V(a, z) 
		&\leq - (1 - \gamma) a + b_0
		= -\rho a - (1 - \gamma - \rho) a + b_0  \nonumber  \\
		&= -\rho V(a,z) - (1 - \gamma - \rho) a + b_0 + \rho m_V.
	\end{align}
	In particular, if $(a,z) \notin B$, then $a > d$ and \eqref{eq:dri_ineq1} implies that 
	\begin{equation}
	    \label{eq:dri_ineq2}
	    \EE_{a,z} V(a_n, Z_n) - V(a, z)  
	    \leq -\rho V(a,z) - (1 - \gamma - \rho) d + b_0 + \rho m_V    
	    \leq -\rho V(a,z).
	\end{equation}
	Let $b:= b_0 + \rho m_V$. Then the stated claim follows from  \eqref{eq:dri_ineq1}--\eqref{eq:dri_ineq2} and the fact that $V$ is bounded on $B$.
%	Let $b:= b_0 + \rho m_V$. Then \eqref{eq:dri_ineq1} and \eqref{eq:dri_ineq2} yields
%		$\EE_{a,z} V(a_n, Z_n) - V(a, z)  
%		\leq -\rho V(a,z) + b \1 \{ (a,z) \in B \}$
%	for all $(a,z) \in \SS$. Since $V$ is bounded on $B$, the stated claim holds.
\end{proof}

\begin{proof}[Proof of Theorem~\ref{t:gs_gnl_ergo_LLN}]
	Claim~(1) can be proved by applying 
	Theorem~19.1.3 (or a combination of Proposition~5.4.5 and Theorem~15.0.1) of 
	\cite{meyn2009markov}. The required conditions in those theorems have been established by
	Lemmas~\ref{lm:psi_irr}, \ref{lm:str_aperi}, \ref{lm:petite_set} and \ref{lm:geo_drift} 
	above.
    Regarding claim (2), Lemmas~\ref{lm:petite_set} and \ref{lm:geo_drift}
    imply that $\EE_{a,z} V(a_n,Z_n) - V(a,z) \leq -n + b \1 \{(a,z) \in B \}$
    for all $(a,z) \in \SS$, where $B:= [0, d] \times \ZZ$ is petite. Since in
    addition $\{(a_t,Z_t)\}$ is $\psi$-irreducible by Lemma~\ref{lm:psi_irr}, 
    Theorem~19.1.2 of \cite{meyn2009markov} implies that
    $\{(a_t,Z_t)\}$ is a positive Harris chain. Claim~(2) then follows from
    Theorem~17.1.7 of \cite{meyn2009markov}.
	
    To verify claim~(3), since we have shown that $\Phi := \{(a_t, Z_t)\}$ is
    positive Harris with stationary distribution $\psi_\infty$, based on
    Theorem~16.1.5 and Theorem~17.5.4 of \cite{meyn2009markov}, it suffices to
    show that $Q$ is $V$-uniformly ergodic.
	Let $\Phi^n$ be the $n$-skeleton of $\Phi$ (see page~62 of \cite{meyn2009markov}). Then $\Phi^n$ 
    is $\psi$-irreducible and aperiodic by Proposition~5.4.5 of
    \cite{meyn2009markov}. Theorem~16.0.1 of \cite{meyn2009markov} and
    Lemmas~\ref{lm:petite_set} and \ref{lm:geo_drift} then imply that $\Phi^n$ is
    $V$-uniformly ergodic, and, there exists $N \in \NN$ such that $||| Q^{nN}
    - 1 \otimes \psi_\infty |||_V < 1$, where $\|\mu\|_V := \sup_{g: |g| \leq V}
    |\int g \diff \mu|$ for $\mu \in \pP(\SS)$ and, 
	for all $t \in \NN$,
	\begin{equation*}
	    ||| Q^{t} - 1 \otimes \psi_\infty |||_V 
            := \sup_{(a,z) \in \SS} 
            \frac{ \| Q^{t}((a,z), \cdot) - \psi_\infty \|_V }{V(a,z)}.
	\end{equation*}
    To show that $Q$ is $V$-uniformly ergodic, by Theorem~16.0.1 of
    \cite{meyn2009markov}, it remains to verify: $||| Q^{t} - 1 \otimes
    \psi_\infty
    |||_V < \infty$ for $t \leq nN$. This obviously holds since, by  the proof
    of Lemma~\ref{lm:geo_drift}, there exist $L_0, L_1 \in \RR$ such that, for
    all $t \in \NN$,
	\begin{align*}
	    ||| Q^{t} - 1 \otimes \psi_\infty |||_V 
	    &\leq \sup_{(a,z) \in \SS} \sup_{ \|f\| \leq V } 
	        \frac{\int |f(a',z')| Q^t ((a,z), \diff (a',z')) }{V(a,z)}
	        + L_0    \\
	    &\leq \sup_{(a,z) \in \SS}
	        \frac{\int V(a',z') Q^t ((a,z), \diff (a',z')) }{V(a,z)}
	        + L_0 
	    \leq L_0 + L_1 < \infty.
	\end{align*}
	Hence, $Q$ is $V$-uniformly ergodic and claim~(3) follows. The proof is now complete.
\end{proof}

\begin{proof}[Proof of Theorem \ref{t:heavy_tail}]
	Take an arbitrarily large constant $k<1$ such that
	\begin{equation*}
		P(\bar z,\bar z)>0 \quad \text{and} \quad \PP_{\bar z} \{ kG(\bar z,\bar z,\hat{\zeta})>1 \}>0,
	\end{equation*}
	which is possible by Assumption \ref{a:wealth_growth} and the definition of $G$ in \eqref{eq:defG}. For this $k$, since $\lim_{a\to\infty}c^*(a,z)/a=\alpha(z)$ and $\ZZ$ is a finite set, we can take $\bar{a}>0$ such that
	\begin{equation*}
		1-\frac{c^*(a,z)}{a} \geq k(1-\alpha(z))
	\end{equation*}
	for all $z\in \ZZ$ and $a\geq \bar{a}$. Multiplying both sides by $R(\hat{z},\hat{\zeta})\geq 0$, it follows from the law of motion \eqref{eq:dyn_sys}, $Y(\hat{z},\hat{\eta})\geq 0$, and the definition of $G$ in \eqref{eq:defG} that for $a\geq \bar{a}$,
	\begin{align*}
		\hat{a}&=R(\hat{z},\hat{\zeta})(a-c^*(a,z)) + Y(\hat{z},\hat{\eta})\\
		&\geq R(\hat{z},\hat{\zeta})(a-c^*(a,z))=R(\hat{z},\hat{\zeta})\left(1-\frac{c^*(a,z)}{a}\right)a\\
		&\geq R(\hat{z},\hat{\zeta})k(1-\alpha(z))a=kG(z,\hat{z},\hat{\zeta})a.
	\end{align*}
	Let $\tilde{A}(z,\hat{z},\hat{\zeta}) := kG(z,\hat{z},\hat{\zeta}) \1 \{ kG(z,\hat{z},\hat{\zeta})>1 \}$. Then for all $z,\hat{z},\hat{\zeta},\hat{\eta}$ and all $a\geq \bar{a}$,
	\begin{equation}\label{eq:a_lb}
		\hat{a} \geq \tilde{A}(z,\hat{z},\hat{\zeta})a.
	\end{equation}
	Start the wealth accumulation process $a_t$ from $a_0\geq \bar{a}$. Consider the following process:
	\begin{equation*}
		S_{t+1}=\tilde{A}(Z_t,Z_{t+1},\zeta_{t+1})S_t,
	\end{equation*}
	where $S_0=a_0$. We now show that $a_t\geq S_t$ with probability one for all $t$ by induction. Since $S_0=a_0$, the case $t=0$ is trivial. Suppose the claim holds up to $t$. Because $a_t\geq 0$ and $S_t$ remains 0 once it becomes 0, without loss of generality we may assume $S_0,\dots,S_t$ are all positive. Hence $\tilde{A}_1,\dots,\tilde{A}_t>0$. By the definition of $\tilde{A}$, we have $\tilde{A}>1$ whenever $\tilde{A}>0$. Therefore
	\begin{equation*}
		S_t=\tilde{A}_t\dotsb \tilde{A}_1S_0\geq S_0=a_0\geq \bar{a}.
	\end{equation*}
	Hence applying \eqref{eq:a_lb}, we get
	\begin{equation*}
		a_{t+1} \geq \tilde{A}(Z_t,Z_{t+1},\zeta_{t+1})a_t \geq \tilde{A}(Z_t,Z_{t+1},\zeta_{t+1})S_t=S_{t+1}.
	\end{equation*}
	
	Now take any $p\in (0,1)$ and let $T$ be a geometric random variable with mean $1/p$ that is independent of everything. Define
	\begin{equation*}
		\tilde{\lambda}(s)=(1-p)r(P\odot M_{\tilde{A}}(s)),
	\end{equation*}
	where $M_{\tilde{A}}(s)$ is as in \eqref{eq:MGF}. 
	Since clearly $A\geq \tilde{A}$ and $p>0$, we have $\lambda>\tilde{\lambda}$. By Lemma~3.1 of \cite{BeareToda-dPL}, $\lambda,\tilde{\lambda}$ are convex, and hence continuous in the interior of their domains. Therefore $\lambda(\kappa)=1$ and $\lambda(s)>1$ for small enough
	$s>\kappa$.  Hence, for any $\epsilon>0$, we can take small enough $p\in (0,1)$ and
	large enough $k<1$ such that
	$\tilde{\lambda}(\kappa)<1<\tilde{\lambda}(\kappa+\epsilon)<\infty$. 
	By Lemma~3.1 of \cite{BeareToda-dPL}, there exists a unique
	$\tilde{\kappa}\in (\kappa,\kappa+\epsilon)$ such that
	$\tilde{\lambda}(\tilde{\kappa})=1$. Theorem~3.4 of \cite{BeareToda-dPL} then implies that
	%Because
	%$\tilde{\lambda}$ is convex (which is easy to show using H\"older's
	%inequality, see \cite{BeareToda-dPL}), there exists a unique
	%$\tilde{\kappa}\in (\kappa,\kappa+\epsilon)$ such that
	%$\tilde{\lambda}(\tilde{\kappa})=1$. Using the Pareto tail results in
	%\cite{BeareToda-dPL}, we obtain
	%
	\begin{equation*}
		\liminf_{a\to\infty}a^{\tilde{\kappa}}\PP_{a_0,z_0} \{ S_T>a \}>0
	\end{equation*}
	for all $(a_0,z_0) \in \SS$. In particular, for any initial $(a_0, z_0) \in \SS$ with $a_0 \geq \bar{a}$,
	\begin{equation}\label{eq:S_T_lb}
		\liminf_{a\to\infty}a^{\kappa+\epsilon}\PP_{a_0,z_0} \{S_T>a \}>0.
	\end{equation}
	Now suppose that we draw $a_0$ from the ergodic distribution. Then $a_t$ has the same distribution as $a_\infty$, and so does $a_T$. Therefore
	\begin{multline}\label{eq:PP_a_inf}
		\PP \{ a_\infty>a \}=\PP \{ a_T>a \}\\
		=\PP \{ a_0<\bar{a} \}\PP \{a_T>a \, | \, a_0<\bar{a} \}+\PP \{ a_0\geq \bar{a} \} \PP \{a_T>a \, | \, a_0\geq \bar{a}\}.
	\end{multline}
	If the ergodic distribution of $\{a_t\}$ has unbounded support, then $\PP \{a_0\geq \bar{a} \}>0$. As we have seen above, conditional on $a_0 \geq \bar{a}$, we have $a_t\geq S_t$ for all $t$. Therefore
	\begin{equation}\label{eq:PP_a_T}
		\liminf_{a\to\infty}a^{\kappa+\epsilon}\PP \{ a_T>a \mid a_0 \geq \bar a \} 
		\geq \liminf_{a\to\infty}a^{\kappa+\epsilon}\PP \{ S_T>a \mid a_0 \geq \bar a\}>0
	\end{equation}
	by \eqref{eq:S_T_lb}, and so \eqref{eq:heavy_tail} follows from \eqref{eq:PP_a_inf} and \eqref{eq:PP_a_T}.
\end{proof}

\bibliographystyle{ecta}

\bibliography{os}

\end{document}